\documentclass[journal ,12]{IEEEtran}
\usepackage{colortbl}
\usepackage[table]{xcolor}
\usepackage{tikz}
\usepackage{algorithm}
\usepackage{algorithmicx}
\usepackage{algpseudocode}
\usepackage{url}
\usepackage{float} 
\usepackage{amsmath,amsfonts,amssymb,mathrsfs,amsthm,amsbsy,graphicx,color,bm,algpseudocode,balance,mathtools,xfrac,nccmath,siunitx,cite}
\usepackage{multirow}

\newtheorem{lem}{Lemma}

\newtheorem{rem}{Remark}
\newcommand{\E}{$\mathcal E$} 
\newcommand{\R}{$\mathcal R$}
\newcommand{\T}{$\mathcal T$}
\newcommand{\Tk}{$\mathcal{T}_k$}
\definecolor{ambe}{rgb}{1.0, 0.49, 0.0}

\usepackage[font=scriptsize,labelfont=bf]{caption}
\usepackage[caption=false]{subfig} 
 \usepackage[utf8]{inputenc}
\usepackage{booktabs, caption, makecell,xspace}

\usepackage{threeparttable}

\usepackage{pgfplots}
\pgfplotsset{compat=newest}
\usetikzlibrary{plotmarks}
\usetikzlibrary{arrows.meta}
\usepgfplotslibrary{patchplots}
\usepackage{grffile}

\newcommand{\bc}{\text{BackCom}\xspace}
\newcommand{\mbc}{\text{MoBC}\xspace} 
\newcommand{\bbc}{\text{BiBC}\xspace}
\newcommand{\abc}{\text{AmBC}\xspace}

 \newtheorem{theorem}{Theorem}

 \theoremstyle{definition}
 \makeatletter
\newcommand{\printfnsymbol}[1]{%
  \textsuperscript{\@fnsymbol{#1}}%
}
\makeatother

\usepackage{pifont}% http://ctan.org/pkg/pifont
\newcommand{\cmark}{\ding{51}}%
\newcommand{\xmark}{\ding{55}}%
\newcommand{\lmark}{\textbf{--}}

\begin{document}
\bstctlcite{IEEEexample:BSTcontrol}

\title{RIS-Assisted Energy Harvesting  Gains for Bistatic Backscatter Networks:  Performance Analysis and RIS Phase Optimization} \vspace{-5mm}

\author{Diluka  Galappaththige\textsuperscript{\textasteriskcentered}\thanks{\printfnsymbol{1}D.  Galappaththige and F. Rezaei contributed equally to this work.}, \IEEEmembership{Member, IEEE}, Fatemeh Rezaei\printfnsymbol{1}, \IEEEmembership{Member, IEEE},   Chintha Tellambura, \IEEEmembership{Fellow, IEEE,}   Sanjeewa Herath, \IEEEmembership{Member, IEEE}
\thanks{D. Galappaththige, F. Rezaei, and C.~Tellambura with the Department of Electrical and Computer Engineering, University of Alberta, Edmonton, AB, T6G 1H9, Canada (e-mail: \{diluka.lg, rezaeidi, ct4\}@ualberta.ca).  \\
\indent S. Herath is with the Huawei Canada, 303 Terry Fox Drive, Suite 400, Ottawa, Ontario K2K 3J1 (e-mail: sanjeewa.herath@huawei.com).}  \vspace{-10mm}}

\maketitle

\begin{abstract}
Inexpensive tags powered by energy harvesting (EH) can realize green (energy-efficient) Internet of Things (IoT) networks. However, tags are vulnerable to energy insecurities, resulting in  poor  communication ranges, activation distances, and data rates.  To overcome these challenges, we  explore the use of a reconfigurable intelligent surface (RIS) for  EH-based IoT networks. The  RIS is deployed to enhance RF power at the tag, improving EH capabilities. We consider linear and non-linear EH models and analyze single-tag and multi-tag scenarios. For single-tag networks, the tag's  maximum received power  and the reader's signal-to-noise ratio  with the optimized RIS phase-shifts are derived.   Key metrics, such as received power, harvested power, achievable rate, outage probability, bit error rate, and diversity order, are also evaluated. The  impact of RIS phase shift quantization errors is also studied. For the multi-tag case,  an algorithm to compute the optimal RIS phase-shifts is developed.  Numerical results and simulations demonstrate significant improvements compared to the benchmarks of  no-RIS case and random RIS-phase design. For instance, our optimal design with a \num{200}-element RIS increases  the  activation distance  by \qty{270}{\percent} and \qty{55}{\percent} compared to those benchmarks. 
In summary, RIS deployment improves the energy autonomy of tags while maintaining the basic tag design intact. 
\end{abstract}


\IEEEpeerreviewmaketitle
\section{Introduction}

\subsection{The Problems with Energy Harvesting  Backscatter Tags}

Parcel tracking using passive electronic tags is one of the many potential applications of the Internet of Things (IoT). The global parcel volume surpassed \num{131} billion in \num{2020}, showing a \qty{27}{\percent} year-over-year increase. In the United States alone, \num{59} million parcels were generated daily in \num{2021}, projected to reach \num{25}-\num{40} billion with a \qty{5}{\percent}-\qty{10}{\percent} annual growth rate from \num{2022}-\num{2027}. Similar growth trends are observed globally. Barcode-based tracking is currently employed, but electronic tag-based tracking offers advantages such as enhanced labor productivity, throughput, warehousing efficiency, and real-time data accuracy for quality control. Tags without batteries are particularly suitable because of their cost-effectiveness and compact size. Their   applications include  medical and healthcare, agriculture, livestock, logistics, retail chains, and passive IoT networks \cite{Yang2010, Fuqaha2015, Huawei}. These tags have low-cost and low-power circuits with limited processing capabilities. They rely on backscatter modulation, a process described in detail in \cite{Diluka2022, Rezaei2023}, where they reflect radio-frequency (RF) signals to communicate with the reader.

Passive tags encounter two main problems related to their reliance on RF energy harvesting (EH) for power: activation failure and energy outage (EO). Activation failure occurs when the tag fails to reach the activation threshold ($P_b$), typically around \qty{-20}{\dB m} \cite{tags}, required to initiate the EH circuitry \cite{Wang2017}. Imperfections in the matching network between the tag's antenna and the EH circuit can cause this failure. The matching network aims to align the complex impedance of the EH circuit with the antenna's impedance, optimizing power transfer and minimizing signal reflections. However, the EH circuit's impedance depends on the incident input power due to nonlinear devices, leading to reduced circuit efficiency with changes in input power. The second problem is EO. Ambient energy sources are unpredictable with  RF power density values as low as  \num{1}$\sim$\qty{100}{\uW/\cm^2}  and varying with distance \cite{Huawei}. As a result, there is a risk of an EO where the tag does not reach the activation threshold. These problems cause  ultra-low  power (\qty{}{n\W}-\qty{}{\uW}),  short communication ranges ($\leq\qty{6}{\m}$), short activation distances, and low data rates ($\leq \qty{1}{bps/\Hz}$). It is clear that all these problems are initiated whenever the incident RF energy is low.  Addressing that issue is the main focus of this paper.    

Backscatter networks can be categorized into three types: monostatic, bistatic, and ambient. In monostatic systems, the reader and emitter are co-located, resulting in doubled path loss \cite{Kimionis2014}. Ambient systems rely on reflecting existing RF signals, which are highly unpredictable. Bistatic systems, on the other hand, offer better support for applications such as warehouses (Fig.\ref{fig:systemGeneral}). These systems deploy dedicated RF emitters, either single or multiple, to provide energy to the tags and enable backscatter modulation. By optimizing the locations of multiple emitters, a larger area can be covered (Fig.\ref{fig:systemGeneral}), maximizing coverage and performance. Dedicated emitters have advantages over ambient signals, including predictability, reduced interference, control over the system, and knowledge of ambient signal parameters \cite{rezaei2020large}. However, the high cost, complexity, and transmit powers associated with dedicated emitters can be problematic.

These considerations motivate the following questions: 1) what is the best way to increase the chance of the incident RF  power on the tag  exceeding $P_b?$ 2) How can that goal be reached    without increasing dedicated RF emitters' cost and energy expenditure?

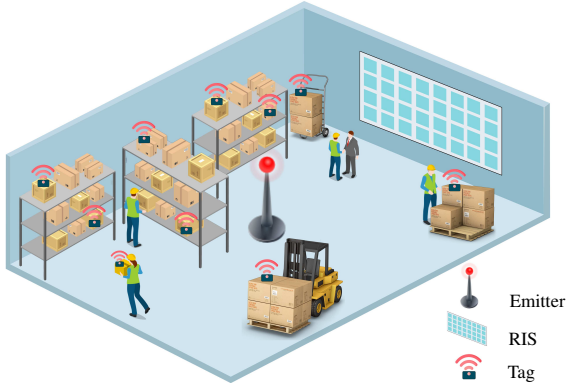
\begin{figure}[!tbp]\vspace{-0mm}
  \centering
        \def\svgwidth{210pt} 
     \fontsize{7}{7}\selectfont 
     \graphicspath{{./Figures/}}
     %% Creator: Inkscape 1.1.1 (3bf5ae0d25, 2021-09-20), www.inkscape.org
%% Accompanies image file 'systemGeneral.eps' (pdf, eps, ps)
\begingroup%
  \makeatletter%
  \providecommand\color[2][]{%
    \errmessage{(Inkscape) Color is used for the text in Inkscape, but the package 'color.sty' is not loaded}%
    \renewcommand\color[2][]{}%
  }%
  \providecommand\transparent[1]{%
    \errmessage{(Inkscape) Transparency is used (non-zero) for the text in Inkscape, but the package 'transparent.sty' is not loaded}%
    \renewcommand\transparent[1]{}%
  }%
  \providecommand\rotatebox[2]{#2}%
  \newcommand*\fsize{\dimexpr\f@size pt\relax}%
  \newcommand*\lineheight[1]{\fontsize{\fsize}{#1\fsize}\selectfont}%
  \ifx\svgwidth\undefined%
    \setlength{\unitlength}{918.72207642bp}%
    \ifx\svgscale\undefined%
      \relax%
    \else%
      \setlength{\unitlength}{\unitlength * \real{\svgscale}}%
    \fi%
  \else%
    \setlength{\unitlength}{\svgwidth}%
  \fi%
  \global\let\svgwidth\undefined%
  \global\let\svgscale\undefined%
  \makeatother%
  \begin{picture}(1,0.69024823)%
    \lineheight{1}%
    \setlength\tabcolsep{0pt}%
    \put(0,0){\includegraphics[width=\unitlength]{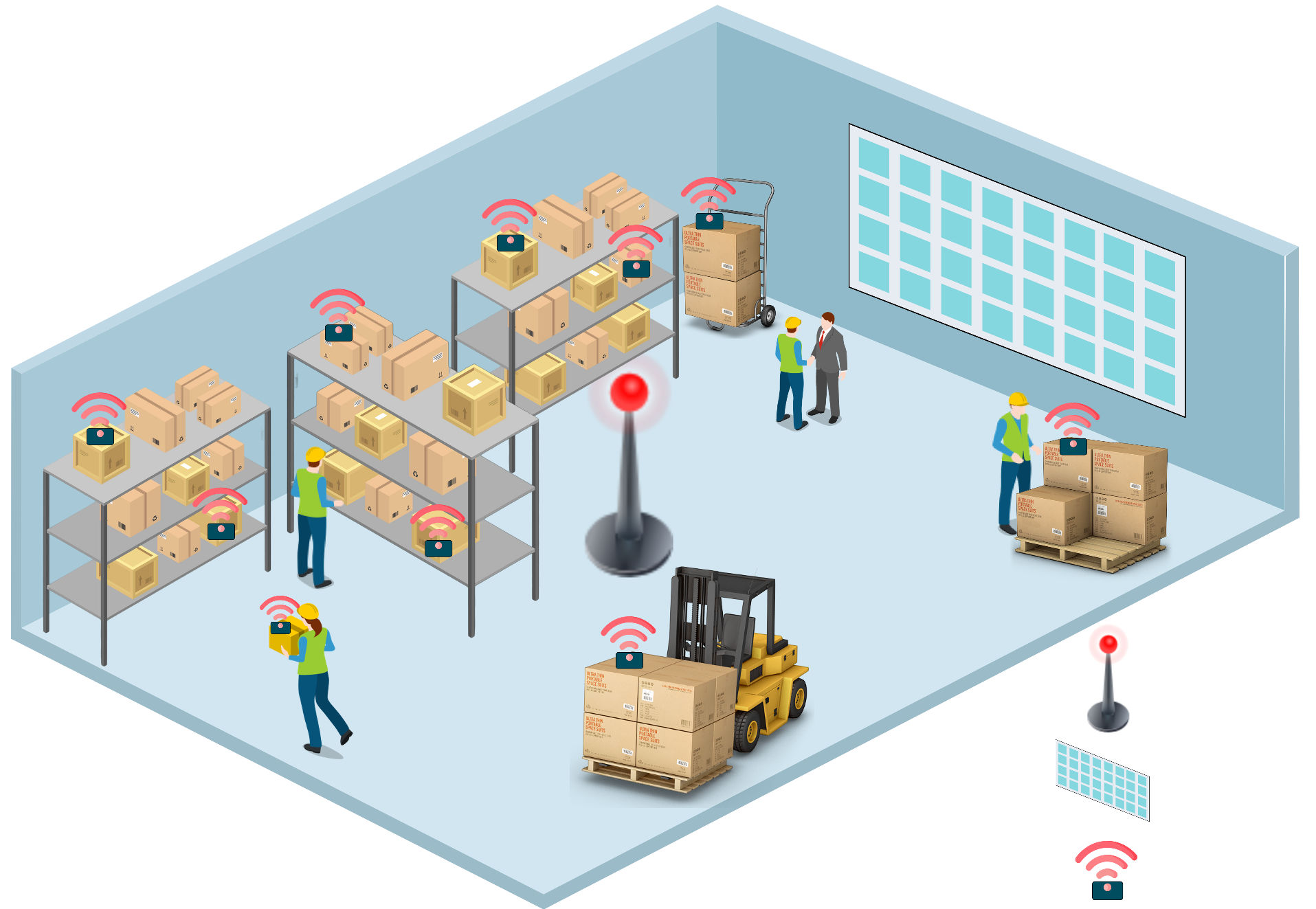}}%
    \put(0.93723736,0.00786808){\makebox(0,0)[t]{\lineheight{1.25}\smash{\begin{tabular}[t]{c}Tag\end{tabular}}}}%
    \put(0.96600089,0.13554467){\makebox(0,0)[t]{\lineheight{1.25}\smash{\begin{tabular}[t]{c}Emitter\end{tabular}}}}%
    \put(0.94036495,0.06800979){\makebox(0,0)[t]{\lineheight{1.25}\smash{\begin{tabular}[t]{c}RIS\end{tabular}}}}%
  \end{picture}%
\endgroup%
 \vspace{-0mm}
     \caption{A warehouse use case of \bbc network.}\label{fig:systemGeneral}
\end{figure}

\begin{figure}[!tbp]\vspace{-0mm}
   \centering
  \def\svgwidth{210pt} 
     \fontsize{7}{6}\selectfont 
     \graphicspath{{./Figures/}}
     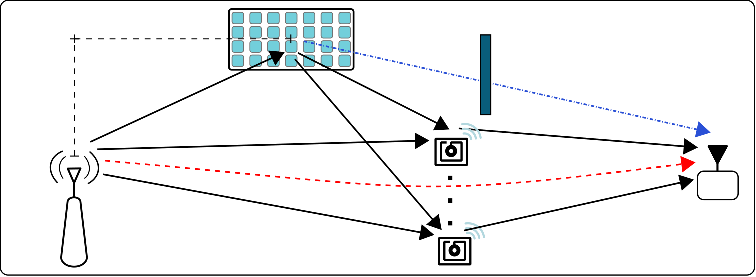 \vspace{-0mm}
     \caption{An RIS-assisted  \bbc. Red lines denote interference signals and  $d_{f_k}, d_{\mathbf{g}_k}, d_{u_k} $ and $d_{\mathbf{h}}$ are respectively \E-\Tk{}, RIS-\Tk{}, \Tk{}-\R, and \E-RIS distances. }\label{fig:system_model}
\end{figure}

\subsection{Existing  Solutions}Various solutions can enhance the performance of passive tags, including the use of multi-antenna configurations \cite{Yang2015}, energy beamforming techniques, RF energy harvesters with improved activation thresholds \cite{Song2022}, and channel coding methods \cite{Rezaei2023, Alevizos2014}. Additionally, active tunnel diodes in the tags have been explored as a solution \cite{Cui2022}. However, these techniques may not be feasible for passive tags due to their limited processing capability, power constraints, and cost limitations. Increasing the transmit power at the emitter or deploying multiple emitters is also not an energy-efficient (green) solution.

To overcome these issues, study   \cite{Jia2022} is the first to  propose the use of a reconfigurable intelligent surface (RIS) in \bbc networks. A RIS is a synthetic surface consisting of numerous passive reflectors, also known as meta-material elements \cite{Poulakis2022, Tapio2021}. These reflectors have the ability to independently modify the characteristics of incoming electromagnetic waves, such as phase, amplitude, frequency, or polarization \cite{Poulakis2022, Tapio2021}. Through real-time adaptation, the reflectors can be programmed to dynamically adjust to the wireless channel conditions. Each reflector is equipped with one or more switches, with a switching frequency of up to \qty{5}{\MHz}, which ensures minimal switching time compared to the channel coherence time \cite{nayeri2018reflectarray}. By dynamically adjusting the reflectors, the RIS controller coordinates their actions to create constructive and destructive interference patterns in the reflected waves as required \cite{Poulakis2022, Tapio2021}.

A RIS consumes a few watts  during the reconfiguration states and much less during idle states, e.g.,  \qty{6}{\mW} per element  for \num{4}-bits resolution phase shifting \cite{Huang2019}. A RIS can provide significant gain even without an amplifier, ranging from \qtyrange{30}{40}{\dB}  relative to the  isotropic radiation  depending on the size of the surface and frequency \cite{Poulakis2022,Tapio2021}. Typically, the size of a single  reflector is much smaller than the signal wavelength ($\lambda$), ranging between $\lambda/10$ and $\lambda/5$ \cite{Tapio2021}. The RIS  integrates with existing networks without modifying the basic network design \cite{Emil2020}.

\subsection{Problem Statement and Contributions}
We address the problem of maximizing the number of passive tags supported in a given area, such as a large warehouse (Fig.\ref{fig:systemGeneral}) while minimizing the number of carrier emitters required. As mentioned before, passive tags have limited activation/communication ranges and data rates, relying on a minimum RF power (activation threshold) for self-activation. Existing solutions like multiple-antenna tags, coding, batteries, multiple emitters, repeaters, and relays have certain limitations. To overcome these challenges, we adopt the approach proposed in \cite{Jia2022} by deploying a RIS  in the \bbc network (Fig.\ref{fig:system_model}). The RIS increases the incident RF power on the tags, allowing them to harvest more energy and improving their reliability and communication range. Optimal placement and configuration of the phase shift elements in the RIS can be determined to achieve maximum gains. Moreover, continuous phase shifts (i.e., infinite levels in $[-\pi,\pi]$) are not viable because more switches per reflector will cost more for the RIS. For  \num{16} phase shifts, $\log_2(16)=4$ switches are needed \cite{Wu2020}. Hence, to keep the size, cost, and power consumption of the RIS down,  phase quantization is essential. 

We propose using a RIS in the \bbc network (Fig. \ref{fig:system_model}) to enhance its energy security and performance. Initially, we consider the case of a single tag and assume a linear energy EH model for the tag. The RIS acts as a reflector, increasing the power delivered to the passive tags by reflecting the RF emitter signal. By optimizing the phase shifts of the RIS, we aim to improve harvested power, communication range, and data rate. Our results provide insights into the improvements in harvested energy, data rate, and reliability achieved through the use of the RIS. Furthermore, to support passive IoT applications involving multiple tags, we develop RIS optimization techniques that aid tag activation while enhancing communication performance.

{
Our study differs from \cite{Jia2022} in several aspects. Firstly, we focus on increasing the energy security of passive tags, resulting in improved activation distance, communication range, and data rates. In contrast, \cite{Jia2022} primarily aims to enhance the communication performance of passive/semi-passive tags. We achieve this by employing a RIS in the source-to-tag links, which enhances the incident energy at the tags. In contrast, \cite{Jia2022} uses a RIS to increase the signal-to-noise ratio (SNR) at the reader, serving a different purpose. Secondly, our study utilizes fully passive tags with reflective modulation schemes such as binary phase-shift keying (BPSK), implemented by switching the tag's impedance between two levels. These passive modulation schemes can be easily extended to higher-order quadrature phase-shift keying (QPSK) and $M$-quadrature amplitude modulation (QAM) \cite{Diluka2022, Rezaei2023}. On the other hand, \cite{Jia2022} employs semi-passive and passive tags with frequency-shift keying (FSK) modulation to enable multiple access. However, implementing FSK requires more from the tags as they need to generate distinct carrier frequencies. Therefore, achieving higher-order FSK may not be feasible with simple tags. Thirdly, in \cite{Jia2022}, the beamformer at the multi-antenna emitter and the RIS phase shifts are optimized to minimize the transmit power at the emitter. In our study, we aim to maximize the sum rate of multiple tags by optimizing RIS phase shifts with a single-antenna emitter that maintains constant transmit power. }

{As well, there are analytical and algorithmic differences between our study and \cite{Jia2022}. While \cite{Jia2022} uses a minorization maximization (MM) algorithm with a semi-passive tag, we propose a closed-form solution for RIS phase shifts that simultaneously maximize the received power of the passive tag and the SNR at the reader. This analytical framework allows us to quantify the benefits of utilizing an RIS for a single passive tag. We evaluate system efficiency, reliability, and error performance through metrics such as harvested power, achievable rate, outage probability, bit error rate (BER), achievable diversity order, and RIS phase shift quantization errors. Additionally, in the multi-tag scenario, \cite{Jia2022} extends the single-tag MM algorithm with minor modifications. In contrast, we employ a fractional programming-based optimization approach that enables multiple access for passive tags. Our approach focuses on maximizing the sum rate of the tags while ensuring tag activation and reducing tag-to-tag interference.  }

Specifically, the  main contributions of this paper  can be summarized as follows: 
 \begin{enumerate}
      \item In the single-tag network, we optimize the RIS phase shifts to maximize the received signal power at the tag. This  has a closed-form solution, allowing us to determine the optimal RIS phase shifts directly. Furthermore, we analyze the impact of the RIS on the EH  process by deriving the average harvested power at the tag. This quantifies the role of  the RIS in enhancing the EH performance.

     \item  To analyze the performance of the system, we require the probability density function (PDF) and cumulative distribution function (CDF) of the optimal SNR at the reader. However, obtaining analytical expressions for these functions appears intractable, even in the single-tag case. To overcome this challenge, we test  Gaussian and Gamma approximations for the PDF and CDF by assuming a Nakagami-$m$ fading. Among the two approximations, we select the Gamma approximation as it provides a more accurate analysis of the harvested power, achievable rate, outage probability, BER, and achievable diversity order.
     
      \item We also investigate the effect of RIS phase shift errors  to quantify the effects  of hardware limitations and imperfect channel estimation in practical scenarios. 
      
     \item We further explore the optimization of the RIS for a multi-tag scenario. To this end, we introduce an iterative algorithm that aims to maximize the achievable sum rate while ensuring that each tag meets its minimum power requirement. This goal  is achieved by optimizing the RIS phase shifts iteratively.

     \item 
Our simulations and numerical results validate the accuracy of the derived analytical results, demonstrating a close agreement between them. These simulations provide valuable insights for system design and assessment, aiding in the practical implementation and evaluation of the proposed approach.
 \end{enumerate}
We show that RIS optimization provides a green solution for energizing and improving the activation distance and communication range of passive tags, enabling massive connectivity without modifications or additional processing at the tag. This makes passive tags suitable for forming the backbone of passive IoT networks. The analytical results presented in this paper offer valuable insights for practical applications. Before proceeding to the technical contributions, we provide an overview of related works in the field of RIS-assisted \bc literature.

\subsection{Previous Contributions on RIS-Assisted \bc}

Existing works often overlook the tag activation requirement and primarily concentrate on utilizing RIS to enhance the tag-reflected signal for improved communication. However, this approach assumes either the use of active tags powered by batteries, where energy harvesting is unnecessary, or tags that can be activated with extremely low power ($P_b = -\infty$), which is not representative of actual tags. These works optimize the RIS without considering the EH constraint, aiming to improve outage, sum rates, error probabilities, or other metrics. References \cite{Zhao2020Backscatter, Khel2022, Nemati2020, Chen2021Backscatter, Solanki2022, Altuwairgi2022, Zuo2021, Chen2021} fall into this category and differ significantly from our contribution. A summary of relevant and recent contributions is provided in Table \ref{tab:summary}.

{\linespread{1.0}
\begin{table*}[t!]\vspace{0mm}
 \caption{Summary of related works.}
    \begin{center}
\begin{threeparttable}
% \scalebox{1.0}{
    \begin{tabular}{|c|c|c|c|c|}
    \hline
        Setup & Reference & Tag type & EH constraint & Contribution\\ \hline \hline
        \multirow{2}{*}{\mbc}  &\cite{Zhao2020Backscatter} & \lmark & \xmark & BER analysis \\ \cline{2-5} 
          & \cite{Khel2022} & \lmark & \xmark & BER, outage, and rate analysis\\
        \hline
        \multirow{7}{*}{\abc}  &\cite{Nemati2020} & Passive & \xmark & BER and coverage analysis\\ \cline{2-5} 
          & \cite{Chen2021Backscatter} & \lmark& \xmark & BER analysis\\
          \cline{2-5} 
          & \cite{Yang2021} & Passive/Semi-passive & \cmark &  Throughput and coverage analysis \\
          \cline{2-5} 
          & \cite{Solanki2022} & Semi-passive & \xmark & Outage and spectral efficiency analysis\\
          \cline{2-5} 
          & \cite{Altuwairgi2022} & Passive & \xmark& Outage and BER analysis\\
          \cline{2-5} 
          & \cite{Galappaththige2023} & Passive& \cmark& Rate maximization\\
          \cline{2-5} &{~~\cite{Chen2021}}$^*$ & \lmark& \xmark & Transmit power minimization\\
          \hline
          \multirow{4}{*}{\bbc}  &\cite{Zuo2021} & \lmark & \xmark& Rate maximization \\ \cline{2-5} 
          & \cite{Jia2022} & Semi-passive/Passive & \cmark& Transmit power minimization\\
          \cline{2-5} 
          & \textbf{This Paper} & Passive & \cmark & \begin{tabular}{@{}c@{}}BER, outage, rate, and harvested power analysis \\ Rate maximization\end{tabular}\\
        \hline
    \end{tabular}%}
    \begin{tablenotes}\footnotesize
\item[$^*$]Symbiotic radio - the reader supports both the primary and the backscatter transmission
\end{tablenotes}
\end{threeparttable}
\end{center}\vspace{-0mm}
    \label{tab:summary}
\end{table*}}
% \linespread{1.5}

In contrast to the works \cite{Zhao2020Backscatter, Khel2022, Nemati2020, Chen2021Backscatter, Solanki2022, Altuwairgi2022, Zuo2021, Chen2021}, a few studies \cite{Yang2021, Jia2022, Galappaththige2023} have explored the use of RIS to enhance the EH potential of backscatter networks. In \cite{Jia2022}, RIS was applied to a \bbc network, serving as the inspiration for our study. The main purpose of the RIS in \cite{Jia2022} was to improve the signal-to-noise ratio (SNR) at the reader. In contrast, we utilize the RIS to enhance the RF signal received by the tags. Earlier, we discussed the differences between \cite{Jia2022} and our study.

Another study  \cite{Yang2021} utilizes a RIS to enable direct tag-to-tag communication when the direct source-to-tag links are blocked. If the tag receives signal power exceeding a certain threshold $P_b$, it first harvests energy and then reflects its data. Otherwise, it operates using a battery. Hence,  \cite{Yang2021} considers the use of  semi-passive tags, which   is fundamentally different from ours. In our previous work \cite{Galappaththige2023}, we also explored the use of a RIS to energize a single tag in an \abc  network, where the RIS is optimized to satisfy the tag's EH constraint while minimizing interference from the ambient source. Once again, the system model in that study is distinct from the one considered here.

% \subsection{Structure and Notations}
\textit{Structure}: This paper is structured as follows. Section \ref{system_modelA} introduces the system and channel model, including the tag and signal models. The analysis of the single-tag system is presented in Section \ref{sec:singleTag}. Section \ref{sec:multi_tag} explores the multi-tag scenario. The analytical results are validated through simulation examples in Section \ref{sim}. Finally, Section \ref{conclusion} provides the conclusion of the paper and outlines potential future research directions.

\textit{Notation}: For random variable $X$,  $f_X(\cdot)$ and $F_X(\cdot)$ denote PDF and  CDF.  $\mathbb{E} \{ \cdot\}$ and $\mathbb{V}\rm{ar}\{ \cdot\}$ denote the  expectation and variance. Lowercase bold  and uppercase bold denote vectors and matrices. $\mathbf{A}^\mathrm{T}$, $\mathbf{A}^\mathrm{H}$ denote transpose and Hermitian transpose of  matrix $\mathbf{A}$.  Moreover, the positive part of real $x$ is denoted by $[x]^+= \max(0,x)$. The gamma function $\Gamma(a)$ is given in~\cite[Eq. (8.310.1)]{Gradshteyn2007}, $D_v(x)$ is parabolic cylinder function \cite[eq.~(9.240)]{Gradshteyn2007}, and $\gamma(n,x)$ is the lower incomplete gamma function \cite[Eq.~(8.350)]{Gradshteyn2007}.  The   complementary error function is $ {\rm{erfc}} (x)$ \cite[eq.~(8.25.4)]{Gradshteyn2007} and  $Q(x) = \frac 1 2  {\rm{erfc}} (x/\sqrt{2})$ is the Gaussian $Q$-function. Finally, $\mathcal{CN}(\bm{\mu},\mathbf R ) $ is a complex Gaussian vector with  mean   $\bm \mu$ and co-variance matrix  $\mathbf R$.

\section{System Model and Preliminaries}\label{system_modelA}

% \begin{figure}[!t]\centering \vspace{-0mm}
%      \def\svgwidth{220pt} 
%      \fontsize{7}{6}\selectfont 
%      \graphicspath{{Figures/}}
%      \input{Figures/system_fig0.eps_tex} 
%      \caption{An RIS-assisted  \bbc. Red lines denote interference signals and  $d_{f_k}, d_{\mathbf{g}_k}, d_{u_k} $ and $d_{\mathbf{h}}$ are respectively \E-\Tk{}, RIS-\Tk{}, \Tk{}-\R, and \E-RIS distances. }\label{fig:system_model}\vspace{-3mm}
% \end{figure}

\subsection{System and Channel Models}\label{sec:sytem_model}
We consider a RIS-assisted \bbc setup consisting of a single-antenna emitter (\E),  $K$-single-antenna passive tags, a single-antenna reader (\R), and a RIS with $N$ passive reflective elements (Fig. \ref{fig:system_model}). We use  \Tk{} to denote the $k$th tag. 
Since tags are batteryless  and entirely rely on EH, we deploy a RIS  to deliver as much RF power as possible. To do this, the RIS controller sets the states of individual reflectors to adjust the phase shifts intelligently  to maximize the received power at each tag. For this,  the controller requires channel state information (CSI) for all the channels in Fig. \ref{fig:system_model}. We assume that 
the controller has a backhaul connection  between  the RIS and \E{}, which can provide all such necessary information \cite{Wu2019}. For brevity, we index  the set of RIS  passive elements   as $\mathcal{N}=\{1,\ldots,N\}$ and the set of tags as $\mathcal{K}=\{1,\ldots,K\}$.

{We consider a block,  flat-fading channel model where  the channel response remains constant over the duration of a block and changes independently from block to block \cite{Long2020}.} During each fading block, the direct channel coefficients in the \E-\Tk{} link and the \Tk-\R{} link are denoted as $f_k$ and $u_k$, respectively. Moreover, the channel coefficient vectors in the \E-RIS link and the RIS-\Tk{} link  are denoted as  $\mathbf{g}_k = [g_{k,1}, \ldots, g_{k,N}]^{\rm{T}} \in \mathbb{C}^{N \times 1}$ and $\mathbf{h} = [h_1, \ldots, h_N]^{\rm{T}} \in \mathbb{C}^{N \times 1}$, respectively. Here, $g_{k,n}$ and $h_n$ for $n \in \mathcal{N}$ and $k \in \mathcal{K}$ are the channels between the $n$th element of the RIS and  \Tk, and  the $n$th element of the RIS and  \E, respectively. All channel envelopes are assumed to be independent Nakagami-$m$ distributed, where $m$ is the shape parameter \cite{Galappaththige2020}. 
A unified representation of all four channels $\mathcal{A} = \{f_k, u_k,h_n,g_{k,n}\}$ is thus given as 
\begin{eqnarray}
a = \alpha_a \exp({j \theta_a}),
\end{eqnarray}
where $ \alpha_a$ is the envelope of the $a$ and  $\theta_a \in [-\pi,\pi]$ is the phase of $a$. The PDF of 
  $ \alpha_a$ is given as
 \begin{eqnarray}\label{Nakagami}
f_{\alpha_a}(x) = \frac{2m_a^{m_a}x^{2m_a-1}}{\Gamma{(m_a)} \Omega_a^{m_a}} \exp{\left(\frac{-m_a x^2}{\Omega_a}\right)} ,
\end{eqnarray} 
where $m_a$ is the shape parameter and $\Omega_a = m_a \zeta_a$ is the scaling parameter, in which $\zeta_a$ accounts for the large-scale fading/path-loss. It should be noted that, 
since the RIS reflective elements are co-located, the large-scale fading parameters are the same for all of them, i.e.,  $\zeta_{g_{k,n}} = \zeta_{g_k} $ and $\zeta_{h_n} =  \zeta_{h}$ for $n \in \mathcal{N}$.

\begin{rem}\label{Rayl}
The Nakagami-$m$  model is versatile to represent a  variety of propagation environments. For instance,  $m = 1$ represents  Rayleigh fading, and $m \rightarrow \infty$   represents the  no fading scenario. Hence, our  performance analysis  covers the special  case of Rayleigh fading channels \cite{Zhao2020Backscatter, Yang2020}. 
\end{rem}

We make the following key assumptions:
\begin{enumerate}
    \item[A1:] 
    Perfect channel state information (CSI) is available for the \E-\Tk, \E-RIS, RIS-\Tk, \Tk-\R, and \E-\R {} channels in the system. However, channel estimation poses challenges due to the passive nature of the RIS and tags, as well as the large number of RIS elements. Sophisticated methods and novel pilot designs are required for accurate estimation. While a comprehensive treatment of channel estimation is beyond the scope of this work, a possible estimation strategy for the system in Fig. \ref{fig:system_model} can be envisioned. Initially, the RIS is set to the non-reflecting state, allowing estimation of the direct \E-\R{} link, the \E-\Tk-\R{} link, and the \E-\Tk{} link using existing methods \cite{Ma2018, rezaei2023timespread}. Subsequently, the cascaded \E-RIS-\R{} channel can be resolved into its \E-RIS and RIS-\R{} components using appropriate techniques \cite{Wei2021}. Finally, the RIS-\Tk{} channels can be obtained using information from other available paths. Thus, this assumption is justified because  emerging techniques can address the challenges associated with channel estimation.
    
    \item[A2:] 
    We assume that the RIS-\R{} link (blue line in Fig. \ref{fig:system_model}) is blocked or negligible because the RIS focuses the reflected beam toward the tags \cite{Emil2022}. The \E-\R{}  signal (red line in Fig. \ref{fig:system_model})  is interference felt at \R.   When \R{} has   CSI  and knows the signal of   \E,  it can  cancel this interference signal \cite{Jia2022, Tao2021, Fasarakis2015}.  If full CSI is unavailable or cancellation is impossible,  our analytical  results serve as fundamental limits or upper bounds of achievable performance.

    \item[A3:] We consider the linear model  for the EH process at \Tk. Even though practical EH circuits have non-linear operating characteristics and an activation threshold, the linear model can be accurate in certain operating regions.

\end{enumerate}
%\vspace{-4mm}
\subsection{Passive Tag Operation}\label{passive_tag}
As mentioned before, we consider the use of passive tags only as they have the potential to enable massive connectivity at a low cost. Thus, \Tk{}   does not have batteries  and entirely relies on the harvested energy from the RF signal transmitted by \E. Specifically, the harvested  energy  powers \Tk's circuit operation and enables it to  backscatter data to \R{} simultaneously \cite{HoangBook2020}. Thus, the amount of harvested energy is the fundamental parameter that ensures the success or failure of  \Tk{}  to communicate. We next set up the basics necessary to estimate the  EH ability of \Tk. 

Let $P_{T,k}$ be the received power at \Tk's antenna. The reflection coefficient of \Tk{} is then given as $\sqrt{\beta}q_m,$ where  $q_m$ is the normalized backscatter symbol selected from a multi-level (${M}$-ary) modulation (i.e., $\vert q_{{m}} \vert^2 \le 1$) and $0 <\beta < 1$ is the fraction of power reflected at \Tk. Therefore,  when \Tk{} reflects the RF signal, $\beta P_{T,k}$ of the RF power is reflected  while the rest $P_{l,k} = (1-\beta)P_{T,k}$ is absorbed for EH purposes.  The EH circuits convert  the RF power $P_{l,k}$ into  direct current (DC) power by using a rectifier. With linear and nonlinear models for the energy harvester, the harvested DC power ($P_{h,k}$) can be defined as \cite{Wang2017}
\begin{eqnarray} \label{eqn:EH_models}
  P_{h,k} = \begin{cases}
    \phi P_{l,k} , & \text{Linear},\\
    \Phi(P_{l,k}) , & \text{Nonlinear},
    \end{cases} 
\end{eqnarray} 
where $\phi$ is the power conversion efficiency, typically measured around  \qtyrange{31.8}{61.4}{\%} at \qty{2.45}{\GHz} \cite{Chen2017eta}, independent of the RF power $P_{l,k}$. 
Moreover, $\Phi(\cdot)$ represents the nonlinear EH function \cite{Wang2017}. Although the linear model offers simplicity, non-linear EH models are also widely used (see \cite{Wang2017,Wang2020} and references therein). However, our solution can be easily extended to such  models as well. We omit the details for brevity.

\subsection{Signal Model}
Node \E{}  transmits an unmodulated  carrier signal $\sqrt{P}s$  to energize tags, where $s$ is the carrier signal satisfying $ \mathbb{E}\{|s|^2\}=1$ and $P$ is the carrier transmission power. In general,  $ s$ should be designed to maximize the EH potential of tags. Thus, waveforms with high peak-to-average power ratio (PAPR) are used  as these increase the RF-to-DC conversion efficiency of the tag's  EH circuit \cite{Clerckx2016}. Hence,  $s$ can be white noise, with  a flat power spectral density and high PAPR \cite{Collado2014}. Other designs include  orthogonal frequency-division multiplexing (OFDM), chaotic, and multisine waveforms \cite{Clerckx2016, Collado2014}.

The signal $\sqrt{P}s$   is received at \Tk{} through the direct channel, $f_k$, and the reflective channels of the RIS, $\mathbf{h}$ and $\mathbf{g}_k$. The received signal at \Tk{} is thus given as\footnote{Note that, since  \Tk{} is a passive device  without active RF components, the noise contributed by \Tk{} can be neglected \cite{Lyu2019}.}%, Xiao2019
\begin{eqnarray}\label{eqn::forward}
    y_k = \sqrt{P} f_k s +  \sqrt{P} \mathbf{g}_k^{\rm{T}} \mathbf{\Theta} \mathbf{h}  s. 
\end{eqnarray}
The first term is due to  the direct path and the second is from the RIS's reflective elements.  In \eqref{eqn::forward}, $\mathbf{\Theta} \in \mathbb{C}^{N \times N}$ is a diagonal matrix that captures the reflection properties (the magnitude of attenuation and the phase shift) of the RIS elements, i.e., $\mathbf{\Theta} = {\rm{diag}} \left(\eta_1 \exp{(j\theta_1)}, \dots, \eta_N \exp{(j\theta_N)} \right)$, where $\eta_n \exp{(j\theta_n)}$ is the reflection coefficient of  the $n$-th RIS  element  with  the magnitude of attenuation $\eta_n$ and the phase shift $\theta_n \in [-\pi,\pi]$. Discrete phase shifts, resulting in a  phase quantization error,  will be treated  in Section \ref{phase_error}. { Moreover, here we only consider a passive RIS  without active amplification, i.e.,  $\eta_k \leq 1, \forall k.$  However, an active RIS that can amplify and reflect incident RF signals, i.e., $\eta_n> 1$, is a potential future extension of this work.} 

Using \eqref{eqn::forward}, the received power at \Tk{} is given as 
\begin{eqnarray}\label{Received_power_k}
    P_{T,k} = P \vert f_k +  \mathbf{g}_k^{\rm{T}} \mathbf{\Theta} \mathbf{h} \vert ^2.
\end{eqnarray}
\Tk{} harvests energy from the received power, $ P_{T,k}$ \eqref{Received_power_k}, and modulates the received signal with
its normalized ${M}$-ary backscatter signal,
$q_k,  \mathbb{E}\{|q_k|^2\} =1$, to be transmitted to \R. The received signal at \R{} is thus given as 
\begin{eqnarray}\label{eqn::backscatter}
    y_R = \sqrt{\beta P}  \sum\nolimits_{k \in \mathcal{K}}  u_k (f_k + \mathbf{g}_k^{\rm{T}} \mathbf{\Theta} \mathbf{h})  s q_k + z,
\end{eqnarray} 
where $z \sim \mathcal{CN}(0,\sigma^2_z)$ is  additive white Gaussian noise (AWGN) with mean \num{0} and variance $\sigma^2_z$.

\vspace{-4mm}
\section{Single-Tag System}\label{sec:singleTag}
Here, we consider $K=1$, i.e., the single-tag setup. The RIS  phase shifts should be optimized to   maximize both the received signal power at \T{} and the achievable rate at \R{}. To this end, the  optimization problem is formulated as follows:
 \begin{subequations} \label{eqn:P1}
	\begin{eqnarray}
		\mathcal{P}_{S,1}: \underset{\mathbf{\Theta} }{\text{maximize}} && {{\rm{log}}_2(1+\gamma)}, \label{eqn:P1_objective}\\
		\text{subject to}  
		&&  (1-\beta) P_T  \geq P_b', \label{eqn:EH_P1}\\
		&& |\bar{\theta}_{n}| \leq 1, \label{eqn:theta_P1}
	\end{eqnarray}	 
\end{subequations} 
where $\bar{\theta}_{n}=\eta_n \exp{(j\theta_n)}$ and $  P_b' \triangleq  P_b/\phi $ is the respective threshold value. Note that this constraint is equivalent to the nonlinear EH case with $P_b'= \Phi^{-1}(P_b)$.
Hence, we adopt the linear EH model  \eqref{eqn:EH_models} due to its tractability.
Here, $P_T$ is the received signal power at \T{} and given as
\begin{eqnarray}\label{eqn::Rx_power}
     P_T =  {P} \left \vert \left(\alpha_f e^{j\theta_f} + \sum\limits_{n \in \mathcal{N}}{\eta_n \alpha_{g_n} \alpha_{h_n} e^{j(\theta_{g_n}+\theta_{h_n} + \theta_n)}} \right) \right \vert^2.
\end{eqnarray} 
When the tag reflects this power level,  the corresponding received SNR at \R{} is given as 
\begin{eqnarray}\label{eqn::received_SNR}
       \gamma &=& \bar{\gamma} \vert u (f + \mathbf{g}^{\rm{T}} \mathbf{\Theta} \mathbf{h}) \vert^2 \nonumber\\
       &=& \bar{\gamma} \left \vert \alpha_u e^{j\theta_u} \! \left(\!\alpha_f e^{j\theta_f}\! +\! \sum\limits_{n \in \mathcal{N}}{\! \eta_n \alpha_{g_n} \alpha_{h_n} e^{j(\theta_{g_n}+\theta_{h_n} + \theta_n)}} \!\right) \right \vert^2\!,\quad 
\end{eqnarray} 
where $\bar{\gamma}= P \beta/\sigma^2_z$.

\subsection{Proposed Solution}
To maximize the received power at \T{} (which will also  maximize the received SNR at \R{} if the interference is negligible), the received signal at \T{} through the RIS (the signal terms inside the summation of \eqref{eqn::Rx_power}), should be constructively added to the signal received from the direct path $f$, by compensating the phase distortion effect of the multipath channel,  i.e., in \eqref{eqn::Rx_power},  $\theta_f=\theta_{g_n}+\theta_{h_n} + \theta_n$.
Thus, the phase shift at the $n$th RIS element  should be adjusted as \cite{Wu2019}
\begin{eqnarray}\label{eqn::phase_adjust}
    \theta^{\star}_n = \max_{-\pi \le\theta_n \le \pi} P_T = \theta_f - (\theta_{g_n}+\theta_{h_n}), ~ n \in \mathcal{N}.\quad 
\end{eqnarray} 
Thereby, the optimal received power at \T{} is derived as
\begin{eqnarray}\label{eqn::Opt_Rx_power}
    P_T^{\star} =  {P} \left \vert \left(\alpha_f  + \sum\nolimits_{n \in \mathcal{N}}{\eta_n \alpha_{g_n} \alpha_{h_n} } \right) \right \vert^2.
\end{eqnarray} 
Similarly, by using \eqref{eqn::phase_adjust}, the optimal received SNR at \R{} is derived as
\begin{eqnarray}\label{eqn::optimal_SNR}
    \gamma^{\star} = \bar{\gamma} \left \vert \alpha_u  \left(\alpha_f + \sum\nolimits_{n \in \mathcal{N}}{\eta_n \alpha_{g_n} \alpha_{h_n} } \right)\right \vert^2.
\end{eqnarray}

Next, we derive the approximate distributions of the  optimal received power \eqref{eqn::Opt_Rx_power} and SNR  \eqref{eqn::optimal_SNR} for the single-tag system. We develop Gaussian and Gamma approximations and compare them briefly, ultimately selecting the Gamma approximation for further analysis. Using it, we analyze various aspects of the single-tag system, including harvested power, achievable rate, outage probability, bit error rate (BER), and diversity order. Additionally, we investigate the impact of phase quantization error on the achievable rate.

\subsection{Average Harvested Power}\label{harvested_power}
We first analyze the impact of the EH process before proceeding to the performance analysis section. 
As mentioned before, \T{} activation depends on  the amount of harvested power, i.e., the sensitivity threshold, $P_b$, which is about \qty{-20}{\dB m} for commercial passive RFID tags \cite{tags}. If \T{} cannot exceed this threshold, it will not be operational to reflect the RF signal to communicate its  data. Thus, we can envision an EO, which will cripple the system.  To understand this scenario, we will analytically determine the average harvested power below.  Note that this derivation is limited to  the linear EH model only because of its simplicity.   However,  the   achievable rate, outage, and BER  depend on the amount of reflected power at \T{} and hence the reflection coefficient, $\alpha$.

The instantaneous harvested power at \T{} is given as $P_h=\phi(1-\beta)P_T^{\star}$, where $\phi$ is in \eqref{eqn:EH_models} and $P_T^{\star}$ is the received power at \T{} \eqref{eqn::Opt_Rx_power}. The average harvested power  is given as
\begin{eqnarray}\label{eqn:avg_HP}
    \bar{P}_h =  \phi
    (1-\beta) \mathbb{E}\{P_T^{\star} \}.
\end{eqnarray} 
To derive $\mathbb{E}\{P_T^{\star} \}$, first, we define $ X \triangleq \sum_{n \in \mathcal{N}}{\eta_n \alpha_{g_n} \alpha_{h_n}} $ and then, $\mathbb{E}\{P_T^{\star} \}$ is given as
\begin{eqnarray}\label{eqn:avg_opt_power}
     \mathbb{E}\{P_T^{\star} \} = P \left(\sigma_{\alpha_f}^{2} + \mu_{\alpha_f}^2 + \sigma_{X}^{2} +\mu_{X}^{2} + 2\mu_{\alpha_f}  \mu_{X} \right),
\end{eqnarray} 
where $\mu_X$ and $\sigma^2_{X}$ are respectively given as
\begin{subequations}\label{eqn::mean_var}
    \begin{eqnarray}%\label{eqn::mean}
        \mu_X &=& \sum \nolimits_{n \in \mathcal{N}}{\eta_n \mu_{\alpha_{g_n}} \mu_{\alpha_{h_n}}},\quad \\ 
        \sigma^2_X &=&  \sum \nolimits_{n \in \mathcal{N}}{\left(\mu^{(2)}_{x_n} -\mu_{x_n}^2\right)},
    \end{eqnarray}
\end{subequations}
in which $ \mu^{(2)}_{x_n} = {\eta_n^2 \mu^{(2)}_{\alpha_{g_n}}\mu^{(2)}_{\alpha_{h_n}}}$, $\mu_{x_n} = {\eta_n \mu_{\alpha_{g_n}} \mu_{\alpha_{h_n}}}$, and
\begin{subequations}
    \begin{eqnarray}\label{nakagami_mean_variance}
        \mu_{\alpha_a}^{(m)} &=&  \frac{\Gamma(m_a+\frac{m}{2})}{\Gamma(m_a)} \left(\frac{\Omega_a}{m_a} \right)^{\frac{m}{2}},\quad \\
        \sigma^2_{\alpha_a} &=& \Omega_a \left( 1-\frac{1}{m_a}  \left(  \frac{\Gamma(m_a+\frac{1}{2})}{\Gamma(m_a)} \right)^2\right),
    \end{eqnarray} 
\end{subequations}
for $a \in  \mathcal{A}$. Besides, $\mu_V = \mathbb{E}\{V\}$,  $\sigma^2_V = \mathbb{V}\text{ar} \{ V\}$ and $ \mu^{(m)}_V = \mathbb{E}\{V^m\}$.

\subsection{Statistics of  the Optimal Received Power and SNR}\label{Statistical_Characterization}

To derive the distribution of the optimal  SNR \eqref{eqn::optimal_SNR}, we  consider $\gamma^{\star} = \bar{\gamma} \Lambda^2$, where $\Lambda  \triangleq \alpha_u  Y$ and $Y  \triangleq \alpha_f + X$. The $\alpha_u, \alpha_f, \alpha_{h_n}$, and $ \alpha_{g_n},  \forall n$ are independent  Nakagami-$m$  variables, and the exact derivations of the PDFs of $Y$, $\Lambda$, $P_T^{\star}$, and $\gamma^{\star}$ are  analytically intractable.  To overcome this challenge, we will develop two approximations, i.e., Gaussian and Gamma, of $Y$ and $\Lambda$ to derive the approximate distribution of $P_T^{\star}$ and $\gamma^{\star}$. We will verify these via simulations. 
%They seem to yield  lower and upper bounds of the exact error rate and outage for a high transmit power regime.

\begin{theorem} \label{Gaussian_app}
\textit{\textbf{Gaussian Approximation}}:
Using the moment matching technique, $Y$ and $\Lambda$ can be approximated as  Gaussian variables with $Y \sim \mathcal{N}(\mu_{Y}, \sigma^2_{Y})$ and $\Lambda \sim \mathcal{N}(\mu_{\Lambda}, \sigma^2_{\Lambda})$, within the interval $[0,\infty)$, where  $\mu_Y = \mu_{\alpha_f} + \mu_X$, $\sigma^2_Y  = \sigma^2_{\alpha_f} + \sigma^2_{X} $, and $\mu_{\Lambda} = \mu_{\alpha_u} \mu_{Y}$, $\sigma^2_{\Lambda} = \mu^{(2)}_{\alpha_u} \mu^{(2)}_{Y} -  \mu_Y^2 \mu_{\alpha_u}^2$.

%$\mu^{(2)}_{\alpha_u} = \sigma^2_{\alpha_u} + \mu_{\alpha_u}^2$.
% \vspace{-3mm}
\end{theorem}
\begin{proof}
    See Appendix \ref{Gaussian}.
\end{proof}

\begin{theorem}  \label{Gamma_app}
\textit{\textbf{Gamma Approximation}}:
$Y$ and $\Lambda$ can  be approximated as  Gamma variables with $ Y \sim \Gamma(k_{Y}, \lambda_{Y})$ and $ {\Lambda} \sim \Gamma(k_{\Lambda}, \lambda_{\Lambda})$, where $k_{C} = \mu^2_{C}/\sigma^2_{C}$ and $\lambda_{C} = \sigma^2_{C}/\mu_{C}$ for $C\in \{Y,\Lambda\}$ are the first and second moments \cite{florescu2014probability,papoulis2002probability}. Hence, the CDF and PDF of $Y$ and $\Lambda$ are  given as
\begin{subequations} 
    \begin{eqnarray}\label{gamma_appro}
        F_{C}(r) &=&  \frac{1}{\Gamma(k_{C})} \gamma(k_{C}, \frac{r}{\lambda_{C}}), \quad \\
        f_{C}(r) &=&  \frac{1}{\Gamma(k_{C})\lambda_{C}^{k_{C}}} r^{k_{C}-1} \exp \left(-\frac{r}{\lambda_{C}}\right), %\label{eqn:pdf_R}
    \end{eqnarray} 
\end{subequations} 
for $r\geq 0$.
\end{theorem}

Since $P_T^{\star} = P Y^2$ and $\gamma^{\star} = \bar{\gamma} \Lambda^2$, we have 
 \begin{subequations}
    \begin{eqnarray}\label{eqn::gammaStar}
        F_{D}(r) &=& F_{d}(\sqrt{r/\bar{d}} ),\quad \\
        f_{D}(r) &=& \frac{1}{(2\sqrt{\bar{d} r})}f_{d}(\sqrt{r/\bar{d}}).
    \end{eqnarray}
\end{subequations}  
Here, for $D=P_T^{\star}$, $d=Y$ and $\bar{d}=P$, and for $D=\gamma^{\star}$, $d=\Lambda$ and $\bar{d}=\bar{\gamma}$.

\begin{rem}\label{Rayleigh}
For Rayleigh fading channels, i.e., $m_a = 1 , a \in \mathcal{A}$,  $\mu_{\alpha_a} = \sqrt{\pi \Omega_a}/2$, $\mu_{\alpha_a}^{(2)} = \Omega_a$, and $\sigma^2_{\alpha_a} = \Omega_a (4- \pi)/4$ \eqref{nakagami_mean_variance}. Besides, for $\eta_n = \eta$, $ \mu_X = \pi/4 N \eta  \sqrt{\Omega_g \Omega_h}$,  $\sigma^2_X =  (1-(\pi/4)^2) N \eta^2  \Omega_g \Omega_h,$
%\begin{subequations}
    % \begin{eqnarray}\label{eqn::mean_Ray}
    %     \mu_X = \pi/4 N \eta  \sqrt{\Omega_g \Omega_h},\quad \text{and} \quad
    %     \sigma^2_X =  (1-(\pi/4)^2) N \eta^2  \Omega_g \Omega_h,
    % \end{eqnarray}%(\Gamma(3/2))^2 
%\end{subequations}
and $\mu_{\Lambda}$ and $ \sigma^2_{\Lambda}$ can be found accordingly as %using \eqref{R_mean}, as
\begin{subequations}
    \begin{eqnarray}\label{eqn::mean_Ray1}
        \mu_{\Lambda} &=& \frac{\pi}{4 }\sqrt{ \Omega_u\Omega_f} + \frac{\pi^{3/2}}{8 } N \eta \sqrt{\Omega_g \Omega_h \Omega_u},\quad \\
        \sigma^2_{\Lambda} &=& \Omega_u \sigma^2_Y +  \frac{(4- \pi)}{4}\Omega_u \mu_Y^2,
    \end{eqnarray}
\end{subequations}
where $ \mu_Y = \sqrt{\pi \Omega_f}/2 + \mu_X$ and $\sigma^2_Y =  \Omega_f (4- \pi)/4 + \sigma^2_X $.
 \end{rem}

\begin{figure}[!tbp]\vspace{-0mm}
  \centering
  \fontsize{14}{14}\selectfont 
    \resizebox{.52\totalheight}{!}{\input{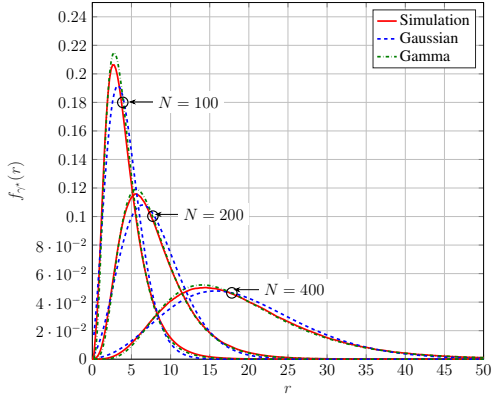}}\vspace{-0mm}
	\caption{The PDF of $\gamma^{\star}$ for $N=\{100,200,400\}$, $P=\qty{10}{\dB m}$, $d_f=\qty{10}{\m}$, $d_u=\qty{5}{\m}$, $d_h=\qty{5}{\m}$, $d_g=\qty{6}{\m}$, $m_f=3$, $m_u=5$, $m_h=3$, $m_g=4$, $\beta=0.6$, and $\eta_n=0.8,\, \forall n$.  }
	\label{fig:PDF} \vspace{-0mm}
\end{figure}

\begin{figure}[!tbp]\vspace{-0mm}
  \centering
  \fontsize{14}{14}\selectfont 
    \resizebox{.495\totalheight}{!}{\input{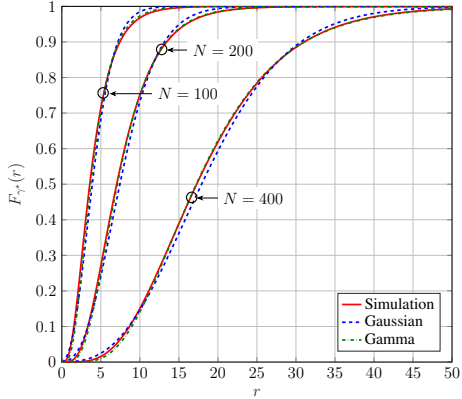}}\vspace{-0mm}
	\caption{The CDF of EO at \T{} as a function of $d_f$ for $P=\qty{20}{\dB m}$,  $d_u=\qty{5}{\m}$, $d_h=\qty{1}{\m}$, $d_g= \sqrt{d_h^2+d_f^2}$, $m_f=m_u=m_h=m_g=3$, $\beta=0.6$, $\phi=0.8$, and $\eta_n=0.8,\, \forall n$. }
	\label{fig:EO_Distance}  \vspace{-0mm}
\end{figure}

% %=======================================================================
% \begin{figure}[!t]\centering \vspace{-3mm}
% 	\includegraphics[width=0.4\textwidth]{Figures/PDF.pdf} \vspace{-1mm}
% 	\caption{The PDF of $\gamma^{\star}$ for $N=\{100,200,400\}$, $P=\qty{10}{\dB m}$, $d_f=\qty{10}{\m}$, $d_u=\qty{5}{\m}$, $d_h=\qty{5}{\m}$, $d_g=\qty{6}{\m}$, $m_f=3$, $m_u=5$, $m_h=3$, $m_g=4$, $\beta=0.6$, and $\eta_n=0.8,\, \forall n$.  }
% 	\label{fig:PDF} \vspace{-2mm}
% \end{figure}
% %=======================================================================

To  verify  these  approximations, we plot the PDF of  $\gamma^{\star}$ in  Fig. \ref{fig:PDF}  for different RIS sizes,   $N$. The Gamma approximation demonstrates better agreement with the exact simulation compared to the Gaussian approximation. As a result, we exclusively utilize it  due to its superior accuracy.

% %=======================================================================
% \begin{figure}[!t]\centering \vspace{-3mm}
% 	\includegraphics[width=0.4\textwidth]{Figures/EO_Distance.pdf}\vspace{-1mm}
% 	\caption{The CDF of EO at \T{} as a function of \E-\T{} distance ($d_f$) for $P=\qty{20}{\dB m}$, $N=\{100,200,400\}$, $d_u=\qty{5}{\m}$, $d_h=\qty{1}{\m}$, $d_g=\sqrt{d_h^2+d_f^2}$, $m_f=m_u=m_h=m_g=3$, $\beta=0.6$, $\phi=0.8$, and $\eta_n=0.8,\, \forall n$. }
% 	\label{fig:EO_Distance}  \vspace{-2mm}
% \end{figure}
% %=======================================================================

\begin{rem}\label{EO}
Fig. \ref{fig:EO_Distance} plots the CDF of EO at \T{} for the cases with and without a RIS. The CDF is evaluated as
\begin{eqnarray}\label{eqn:HP_CDF}
    F_{\rm{EO}}(P_b)= {\rm{Pr}} \left( (1-\beta) P_T^{\star} \leq P_b' \right).
\end{eqnarray}
%where $P_b$ is the tag activation sensitivity, i.e., $P_b= \qty{-20}{\dB m}$.
When there is no RIS to power up the tag, the range is just $\sim \qty{6}{\m}$, and beyond \qty{6}{\m}, the tag goes to an EO. However, a RIS can lower the EO and achieve higher ranges than without a RIS case. For example, a RIS with \num{100} and \num{400} reflecting elements extends the tag activation range to \qty{13}{\m} and \qty{34}{\m} metres, respectively. Therefore, a RIS can prevent the risk of EO at the tag. 
% \commD{We thus present the rest of the analysis assuming that the tag is activated and functioning without an EO.} 
 \end{rem}

\subsection{Average Achievable Rate }\label{Acheivable_rate}
The achievable rate  describes the bit rate  a network  can reliably transmit under a long-term, stable fading regime.   This measure provides valuable insights into the use of specific coding schemes.  In the context of \bbc, it may help to identify feasible applications.  The average achievable rate is the mean Shannon capacity given by   
 \begin{eqnarray}\label{eqn::average_rate}
\mathcal{R} = \mathbb{E} \{\log_2 \left(1+ \gamma^{\star} \right) \}.
\end{eqnarray}
The exact derivation of the average achievable rate given in \eqref{eqn::average_rate} seems intractable. However, by invoking Jensen’s inequality,  the tight upper and lower bounds are given as \cite{Zhang2014}
\begin{eqnarray}\label{eqn::upper_lower}
    \mathcal{R}_{\rm{lb}}  \le  \mathcal{R}\le \mathcal{R}_{\rm{ub}},
\end{eqnarray}
where $\mathcal{R}_{\rm{lb}}$ and $\mathcal{R}_{\rm{ub}}$ are respectively defined as
\begin{subequations}
   \begin{eqnarray}\label{eqn::upperLower}
        \mathcal{R}_{\rm{lb}} &=& \log_2 \left(1+ \left [\mathbb{E} \left \{ {1}/{\gamma^{\star}} \right  \} \right ]^{-1} \right), \quad \\
        \mathcal{R}_{\rm{ub}} &=&\log_2 \left(1+ \mathbb{E} \left \{ {\gamma^{\star}} \right \} \right).
    \end{eqnarray}
\end{subequations}
where $\mathbb{E} \left \{ {\gamma^{\star}} \right \} = \bar{\gamma} \mu^{(2)}_{\Lambda}$, and 
\begin{eqnarray}\label{eqn::expectation}
    \mathbb{E} \left \{ {1}/{\gamma^{\star}} \right  \} = \frac{1}{\mathbb{E} \left \{ {\gamma^{\star}} \right \}} + \frac{\sigma^2_{{\gamma}^{\star}}}{\left [ \mathbb{E} \left \{ {\gamma^{\star}} \right \} \right ]^3},
\end{eqnarray}
in which $\sigma^2_{{\gamma}^{\star}} = \bar{\gamma}^2 \sigma^2_{{\Lambda}^2}$. To calculate $\sigma^2_{{\gamma}^{\star}}$, we need to calculate  $\sigma^2_{{\Lambda}^2} = \mu^{(4)}_{\Lambda} - [\mu^{(2)}_{\Lambda}]^2 $ of $\Lambda$. Specifically, when $\Lambda$  is approximated with Gamma distribution, $ {\Lambda} \sim \Gamma(k_{\Lambda}, \lambda_{\Lambda})$, \cite{papoulis2002probability}
 \begin{eqnarray}\label{eqn::r4gamma}
  \mu^{(m)}_{\Lambda} = k_{\Lambda}^m \frac{\Gamma(m+k_{\Lambda})}{\Gamma(k_{\Lambda})}.
\end{eqnarray}

\begin{rem}
 {To gain further insights, we let the number of RIS elements grow without a bound. We observe that the transmit power can be scaled inversely proportional to the square of the number of IRS elements in this operating regime, i.e., $\lim_{N\rightarrow \infty} P = P_{A}/N^2$. This is the well-known squared power gain provided by the RIS \cite{Emil2020PowerScaling}. Hence, the lower and upper rate bounds can be shown to approach an asymptotic limit using this transmit power scaling law as   (Appendix \ref{asym_rate})
\begin{eqnarray}
    \lim_{N \rightarrow \infty}  \mathcal{R}_{\rm{lb}}= \mathcal{R}_{\infty}, \quad \text{and} \quad \lim_{N \rightarrow \infty}  \mathcal{R}_{\rm{ub}}= \mathcal{R}_{\infty},
\end{eqnarray}
where $\mathcal{R}_{\infty}$ is given in \eqref{eqn_asym_rate} 
\begin{figure*}
\begin{eqnarray}\label{eqn_asym_rate}
    \mathcal{R}_{\infty} = \log_2 \left(1+ \bar{\gamma}_A (\sigma_{\alpha_u}^2+ \mu_{\alpha_u}^2)\left( \eta \frac{\Gamma(m_g+1/2) \Gamma(m_h+1/2)}{\Gamma(m_g) \Gamma(m_h)} \sqrt{\frac{\Omega_g \Omega_h}{m_g m_h}}\right)^2 \right)
\end{eqnarray}
\hrulefill
\end{figure*}
with $\bar{\gamma}_A= \beta P_A/\sigma_z^2$ and $\eta = \eta_n, \forall n$.}
\end{rem}

\subsection{Outage Probability}\label{outage}
The EO at \T{} and rate/SNR outage at \R{} both affect the ability to deliver information reliably. An EO occurs when the received signal power falls below \T's activation threshold ($P_b$). In contrast, the SNR outage probability is the probability that the instantaneous SNR falls below a threshold ($\gamma_{\rm{th}}$). Mathematically,   the outage is expressed as  
\begin{eqnarray}\label{eqn::outage}
    P_{\rm{out}} = P_{\rm{out}}^{\rm{EO}} + \left(1-P_{\rm{out}}^{\rm{EO}}\right) P_{\rm{out}}^{\rm{SNR}},
\end{eqnarray}
where $P_{\rm{out}}^{\rm{EO}} = {\rm{Pr}} \{ (1-\beta) P_T^{\star} \leq P_b \}$ and $P_{\rm{out}}^{\rm{SNR}} =\rm{Pr} \{\gamma \le  \gamma_{\rm{th}} \}$.
A tight approximation to the outage  \eqref{eqn::outage} is obtained as
\begin{eqnarray}\label{eqn::outage2}
    P_{\rm{out}} \approx F_{P_T^{\star}}( \gamma_{\rm{th}}) + \left(1-F_{P_T^{\star}}( \gamma_{\rm{th}}) \right) F_{\gamma^{\star}}( \gamma_{\rm{th}}),
\end{eqnarray}
where $F_{\gamma^{\star}}(r)$ and $F_{P_T^{\star}}(r)$ are given in (\ref{eqn::gammaStar}).

\subsection{Average Bit/Symbol Error Rate}

The average BER  is  the ratio between the  number of bits received in error and the number of bits transmitted.   This  measure  is widely used to assess the reliability and quality of wireless links. In this section, we derive the BER of  the system for  BPSK modulation, a simple yet  robust modulation scheme. The  conditional error probability of BPSK is given as $ {P_{\rm{BER}}}\left( \gamma \right) = \lambda {Q}( {\sqrt {\nu \gamma} })$,  with $\lambda=1$ and $\nu=2$ \cite{Proakis2007}. The average BER of the system with  BPSK, using Gamma distribution \eqref{gamma_appro}, is thus  given as (Appendix \ref{BER_average})

\begin{eqnarray}\label{eqn::BER_gamma}
        \nonumber \bar{P}_{\rm{BER}}^{\rm{Gamma}} &=& A_1  (2\bar{A})^{-\frac{k_{\Lambda}}{2}} \Gamma(k_{\Lambda}) \exp{\left(\frac{\hat{B}_1^2}{8\bar{A}} \right)}  D_{-k_{\Lambda}}\left(\frac{\hat{B}_1^2}{\sqrt{2\bar{A}}} \right),
    \end{eqnarray}
 where $\bar{A} = A \nu \bar{\gamma}, \bar{B} = B \sqrt{\nu \bar{\gamma}} $, $A_1 = \exp{(-C)}/(\Gamma(k_{\Lambda}) \lambda_{\Lambda}^k)$, $\hat{B}_1 = \bar{B}+1/\lambda_{\Lambda}$, in which $A= 0.3842, B = 0.7640,$ and $C= 0.6964$. Note  that  the BER analysis above can be extended for higher-order modulation schemes such as QPSK and $M$-QAM \cite{Yunxia2005, Tellambura1996}. For brevity, such analyses are omitted here.  

\subsection{Achievable Diversity Order}\label{diversity}
The diversity order reveals how the BER  or SNR outage probability decays with the  SNR. It is defined as the negative slope of those measures  in the high SNR regime as follows: 
\begin{eqnarray}
    G_d= -\lim_{\bar{\gamma} \rightarrow \infty} \frac{\log(P_{\rm{out}}^{\rm{SNR}})}{\log(\bar{\gamma})} = -\lim_{\bar{\gamma} \rightarrow \infty} \frac{\log(\bar{P}_{\rm{BER}})}{\log(\bar{\gamma})}.
\end{eqnarray}
We next  derive   the  outage probability and   BER in the high SNR domain. 
We  use the  first-order polynomial expansion of the PDF approximations.  

\begin{lem}
The asymptotic  Outage Probability and  BER are respectively given as
 \begin{eqnarray}\label{outage_asym}
         P_{\rm{out}}^{\infty} = \frac{1}{\Gamma(k_{\Lambda}+1)\lambda_{\Lambda}^{k_{\Lambda}}}  \left( \frac{\gamma_{\rm{th}}}{\bar{\gamma}}\right)^{{k_{\Lambda}}/{2}} + \mathcal{O}\left(\bar{\gamma}^{-({k_{\Lambda}}/{2}+1)}\right),
    \end{eqnarray}
and 
 \begin{eqnarray}\label{high_snr_ber}
        \bar{P}_{\rm{BER}}^{\infty} &=&  ({\lambda_E} \bar{\gamma})^{-k_{\Lambda}/2} +\mathcal{O}\left(\bar{\gamma}^{-(k_{\Lambda}/2+1)} \right),
    \end{eqnarray} 
    where $\lambda_E =\nu \left[\left(2^{k_{\Lambda}/2}/3+(3/2)^{k_{\Lambda}/2} \right) C_1 \right]^{2/k_{\Lambda}}$ and $C_1={\lambda} \Gamma(k_{\Lambda}/2)/(8\Gamma(k_{\Lambda}) \lambda_{\Lambda}^{k_{\Lambda}} )$.

    % \vspace{-3mm}
\end{lem}
\begin{proof}
    See Appendix \ref{Asymptotic_lem}.
\end{proof}
Using \eqref{outage_asym}, the diversity order $G_d$ and the coding gain, $O_c$, are respectively given as $G_d=k_{\Lambda}/2$, and  $O_c={1}/({\Gamma(k_{\Lambda}+1)\lambda_{\Lambda}^{k_{\Lambda}}})$ \cite{Zhengdao2003}. Besides, from \eqref{high_snr_ber} the array gain becomes  $G_a=\lambda_E$.

\begin{rem}
According to Section \ref{Statistical_Characterization}, ${k_{\Lambda}} =  {\mu^2_{\Lambda}}/{\sigma^2_{\Lambda}}$ and $\lambda_{\Lambda} = \sigma^2_{\Lambda}/\mu_{\Lambda}$. Thus, the diversity order $G_d$  can be  given in \eqref{Eqn:diversity_inf}, 
\begin{figure*}
\begin{eqnarray}\label{Eqn:diversity_inf}
    G_d = \frac{1}{2}\frac{ \Upsilon_{u}^2 \zeta_{u} \left( \Upsilon_{f} \zeta_{f}^{1/2}+  N \hat{\Upsilon}  \zeta_g^{1/2} \zeta_h^{1/2}\right)^2}{\zeta_u \left(\bar{\Upsilon}_u +\Upsilon_u^2\right)  \left(\bar{\Upsilon}_f \zeta_f +  N {\Upsilon}' \zeta_g \zeta_h \right)  + \bar{\Upsilon}_u \zeta_u  \left( \Upsilon_f \zeta_f^{1/2}+  N \hat{\Upsilon}  \zeta_g^{1/2} \zeta_h^{1/2} \right)^2},
\end{eqnarray}
% \vspace{-3mm}
\hrulefill
\end{figure*}
where we assume that $\eta_n = \eta, n \in \mathcal{N}$. In \eqref{Eqn:diversity_inf}, we define $\mu_{\alpha_a}\triangleq \Upsilon_a \zeta_a^{1/2}$,  and $\sigma^2_{\alpha_a} \triangleq \bar{\Upsilon}_a \zeta_a$, where $\Upsilon_a= {\Gamma(m_a+{1}/{2})}/{\Gamma(m_a)}$ and $\bar{\Upsilon}_a= m_a- \Upsilon_a^2$ \eqref{nakagami_mean_variance}. Besides, for $\eta_n = \eta, n \in \mathcal{N}$,  $\mu_X = N \hat{\Upsilon}  \zeta_g^{1/2} \zeta_h^{1/2}$, $\sigma^2_{X} = N {\Upsilon}' \zeta_g \zeta_h$, where $\hat{\Upsilon} \triangleq \eta \Upsilon_g \bar{\Upsilon}_h$ and ${\Upsilon}' \triangleq \eta^2 \bar{\Upsilon}_h \bar{\Upsilon}_g $.

To gain further insight, we let the number of RIS elements increase to a large limit, i.e., $ N \rightarrow \infty$. After several  mathematical manipulations,  we have\footnote{Note that $G_d$ is derived using the Gamma approximation and letting $N$ grow large. Hence, $G_d$  may not be exact. However, since the exact derivation is  mathematically intractable, this approximation provides valuable insights.}
\begin{eqnarray}\label{Eqn:diversity_final}
    \lim_{N \rightarrow \infty} G_d = \frac{1}{2} \frac{ \Upsilon_{u}^2  }{ \bar{\Upsilon}_u} =\frac{1}{2} \left(m_u \left(\frac{\Gamma(m_u)}{\Gamma(m_u+\frac{1}{2})}\right)^2 - 1\right)^{-1}. \quad
\end{eqnarray} 
From \eqref{Eqn:diversity_final}, unexpectedly, we  observe that the diversity order of  the proposed RIS-assisted system (Fig. \ref{fig:system_model})  is constant in the high regime of $N$ and solely depends on the \T-\R{}  channel parameters. This behavior is because \T{} acts as a keyhole/pinhole \cite{Diluka2022}. Hence, regardless of the number of RIS elements or antennas of \E{}, the rank of combined channels \E-\T{} and \E-RIS-\T{} becomes one.  Therefore, this combined channel does not contribute to the diversity order because of  keyhole channel properties \cite{Boyer2013}.   

The coding gain and  array gain are similarly derived  as
\begin{eqnarray}
    \lim_{N \rightarrow \infty} O_c &=& \frac{1}{\Gamma \left(\frac{ \Upsilon_{u}^2}{ \bar{\Upsilon}_u}+1\right)} 
    \left ({N \hat{\Upsilon} \sqrt{\zeta_u \zeta_g  \zeta_h} \bar{\Upsilon}_u {\Upsilon}_u^{-1}}\right) ^{- \frac{ \Upsilon_{u}^2  }{ \bar{\Upsilon}_u}},  \label{eqn:codin_gain} \\
    \lim_{N \rightarrow \infty} G_a &=& \lambda_{\bar{E}} {\left ({N \hat{\Upsilon} \sqrt{\zeta_u \zeta_g  \zeta_h} \bar{\Upsilon}_u {\Upsilon}_u^{-1}}\right)}^{-2},\label{eqn:array_gain}
\end{eqnarray}
where $\lambda_{\bar{E}} = \nu \left[\left(2^{{ \Upsilon_{u}^2}/{ (2\bar{\Upsilon}_u)}}/3+(3/2)^{{ \Upsilon_{u}^2}/{( 2\bar{\Upsilon}_u)}} \right) C_{\bar{1}} \right]^{ \frac{2 \Upsilon_{u}^2  }{ \bar{\Upsilon}_u} } $, and  $C_{\bar{1}} = \lambda \Gamma({ \Upsilon_{u}^2}/{ (2\bar{\Upsilon}_u)})/(8\Gamma({ \Upsilon_{u}^2}/{\bar{\Upsilon}_u}))$.

From \eqref{eqn:codin_gain} and \eqref{eqn:array_gain},  while the diversity order is independent of the number of RIS elements, the coding gain and the array gain are both functions of it - see    Fig. \ref{fig:outage_m_N_WO} and Fig. \ref{fig:BER_m_N_WO}.

\end{rem}

\subsection{Effect of RIS Phase Quantization Errors}\label{phase_error}
Due  to  hardware limitations and  imperfect CSI, adopting continuous RIS  phase shifts over $[-\pi ,\pi]$ might be infeasible. Thus, phase shifts would be quantized  to a set of  discrete values  \cite{Wu2020}, i.e.,  each RIS element uses only a  finite number of phase shifts.  Accordingly,  the discrete phase of  $n$-th RIS element is  given as $\hat{\theta}^{\star} = \kappa \pi/2^D$, where  $D$ represents the number of
quantization bits, $\kappa = \min_{\varrho \in \{0, \pm1, \ldots, \pm 2^{D-1}  \}} \vert{\theta}^{\star} - \varrho \pi/2^D\vert$, and ${\theta}^{\star}$ is the optimal phase shift, given in \eqref{eqn::phase_adjust}. Thus, the quantization error is  $\epsilon_n = {\theta}^{\star}_n-\hat{\theta}^{\star}_n, n \in \mathcal{N},$ which can be shown to be uniformly distributed for large quantization levels, i.e., $\epsilon_n \sim \mathcal{U}[-\tau,\tau)$ and $\tau = \pi/2^D$ \cite{Haykin2009,Galappaththige2022}. Therefore, the optimal received SNR $\gamma^{\star}$ given in \eqref{eqn::optimal_SNR} with the discrete phase shifts may be expressed  as
\begin{eqnarray}\label{eqn::optimal_SNR_phase}
     \gamma^{\star} &=& \bar{\gamma} \left \vert \alpha_u  \left(\alpha_f + \sum\nolimits_{n \in \mathcal{N}}{\eta_n \alpha_{g_n} \alpha_{h_n}  e^{j\epsilon_n}} \right)\right \vert^2 \nonumber\\
     &=& \bar{\gamma} \alpha_u^2 ((\alpha_f+X_R)^2 + X_I^2),
\end{eqnarray}
where $X_{R} \triangleq \sum_{n \in \mathcal{N}}{\eta_n \alpha_{g_n} \alpha_{h_n} \cos(\epsilon_n)}$ and $X_I \triangleq \sum_{n \in \mathcal{N}}{\eta_n \alpha_{g_n} \alpha_{h_n} \sin(\epsilon_n)}$. 

We  investigate the effect of phase quantization errors at the RIS on the average achievable rate. According to Section \ref{Acheivable_rate},  we can  derive an upper bound for the average achievable rate, given in (\ref{eqn::upperLower}), with discrete phase shifts as follows:
\begin{eqnarray}\label{eqn::upper_phase}
      \bar{\mathcal{R}}_{\rm{ub}} \!= \log_2  \!\left(1 \!+\! \bar{\gamma}  \mu_{\alpha_u}^{(2)} \left( \mu_{\alpha_f}^{(2)} \!+\!  \mu_{X_R}^{(2)} \!+\! 2 \mu_{\alpha_f}\mu_{X_R} \!+\!  \mu_{X_I}^{(2)} \right) \right),
\end{eqnarray}
where  $\mu_{X_R}^{(2)} = \sigma^2_{X_R} + \mu_{X_R}^{2}$, $\mu_{X_I}^{(2)} = \sigma^2_{X_I} + \mu_{X_I}^{2}$, and 
\begin{subequations}
    \begin{eqnarray}\label{eqn::xi}
        \mu_{X_R} &=& \sum \nolimits_{n \in \mathcal{N}}{\eta_n \mu_{\alpha_{g_n}} \mu_{\alpha_{h_n}}} \frac{\sin(\tau)}{\tau},\quad \\
        \sigma^2_{X_R}&=& \sum \nolimits_{n \in \mathcal{N}}{ \eta_n^2 \mu_{\alpha_{g_n}}^{(2)}  \mu_{\alpha_{h_n}}^{(2)} \left(\frac{1}{2}+ \frac{\sin\left( {2\tau}\right)}{{4\tau}}\right)} -  \mu_{X_R}^2, \qquad\\
        \sigma^2_{X_I} &=& \sum \nolimits_{n \in \mathcal{N}}{\eta_n^2  \mu_{\alpha_{g_n}}^{(2)}  \mu_{\alpha_{h_n}}^{(2)}  \left(\frac{1}{2}- \frac{\sin( {2\tau})}{{4\tau}}\right)},
    \end{eqnarray} 
\end{subequations}
where   $\mu_{X_I} = 0$, and $\mu_{\alpha_a}^{(m)}, a \in \mathcal{A}$  is given in (\ref{nakagami_mean_variance}).

\begin{rem}
If a RIS is placed in the \T-\R{}  link of a \bbc setup, we can investigate the network's performance following a similar approach for the setup with RIS in the \E-\T{} link. However, when both the \E-\T{} and \T-\R{}  links use RIS, the Gaussian and Gamma approximations may not match well with the exact PDF and CDF of the received SNR. We can still derive the upper and lower bounds of the average achievable rate using a similar method given in Section \ref{Acheivable_rate}.  
\end{rem} 

\section{Multi-Tag System}\label{sec:multi_tag}
% ============================ %
This section  develops  the analysis of  multiple ($K>1$) single-antenna tags in Fig. \ref{fig:system_model}. Here, by optimizing the RIS phase shifts, we simultaneously maximize the received signal power at \Tk{} and the achievable sum rate at \R{}. To this end, the problem is formulated as follows:
\begin{subequations} \label{eqn:P1}
	\begin{eqnarray}
		\mathcal{P}_{M,1}: \underset{\mathbf{\Theta} }{\text{maximize}} && \sum \limits_{k \in \mathcal{K}}{{\rm{log}}_2(1+\gamma_k)}, \label{eqn:P1_objective}\\
		\text{subject to}  
		&&  (1-\beta) P_{T,k} \geq P_b', \label{eqn:EH_P1}\\
		&& |\bar{\theta}_{n}| \leq 1, \label{eqn:theta_P1}
	\end{eqnarray}	 
\end{subequations} 
where $P_{T,k}$ is given in \eqref{Received_power_k}, and $\gamma_k$ is the signal-to-interference-plus-noise ratio (SINR)  of \T{}$_k$ at  \R{}. Using \eqref{eqn::backscatter}, $\gamma_k$  is obtained as
\begin{eqnarray}\label{eqn:SINR_Tk}
	\gamma_k = 	\frac{\beta P\vert u_k (f_k + \mathbf{g}_k^{\rm{T}} \mathbf{\Theta} \mathbf{h}) \vert^2}{\sum \nolimits_{i \in \mathcal{K}/k }{\beta P\vert u_i (f_i + \mathbf{g}_i^{\rm{T}} \mathbf{\Theta} \mathbf{h}) \vert^2}+ \sigma^2_z}.
	\end{eqnarray}	 
In \eqref{eqn:P1}, the constraint \eqref{eqn:EH_P1} guarantees the minimum power  for \T{}$_k$. 
% Here, we adopt the linear EH model  \eqref{eqn:EH_models} due to its tractability, i.e., $P_b'=P_b/\phi$ (Section \ref{passive_tag}).

Because of the  non-convex objective function and the constraint,   $\mathcal{P}_{M,1}$ is not amenable to popular  convex algorithms.  Therefore, we develop a solution based  on fractional programming.  

\subsection{Proposed Solution}
First, we define $\mathbf{a}_k\triangleq{\rm{diag}}(u_k\mathbf{g}_k^{\rm{T}}) \mathbf{h}$ and  $ b_k \triangleq u_k f_k$. Thereby, the SINR  \eqref{eqn:SINR_Tk} is rearranged as 
\begin{eqnarray}\label{eqn:SINR_Tk_theta}
	\gamma_k = 	\frac{\beta P \vert b_k + \boldsymbol{\theta}^{\rm{H}} \mathbf{a}_k \vert^2}{\sum \nolimits_{i \in \mathcal{K}/k }{\beta P\vert  b_i + \boldsymbol{\theta}^{\rm{H}} \mathbf{a}_i\vert^2}+ \sigma^2_z},
\end{eqnarray}	
where $\boldsymbol{\theta}= [\bar{\theta}_1,\ldots,\bar{\theta}_N]^{\rm{T}}$. Then, $\mathcal{P}_{M,1}$ can be treated as a fractional programming problem \cite{Shen2018, Galappaththige2023}. We next apply a quadratic transform to the objective function of $\mathcal{P}_{M,1}$ as
\begin{eqnarray}\label{eqn:obj_theta}
	\!\!\!f(\boldsymbol{\theta},\lambda) \!&=&\!  \sum\limits_{k \in \mathcal{K}}\! {\rm{log}_2} \! \left(\! 1 \!+\! 2 \lambda_k \sqrt{\beta P} {\rm{Re}}\left\{b_{k}\!+\!\boldsymbol{\theta}^{\rm{H}} \mathbf{a}_{k}\right\} \right. \nonumber\\
  &&\left. \!-\!\lambda_k^2 \!\left( \!\beta P \!\sum\limits_{i\in \mathcal{K}/k} \!\vert b_{i}+ \boldsymbol{\theta}^{\rm{H}} \mathbf{a}_{i} \vert^2 \!+\! \sigma_{z}^2  \right)\right)\!,
\end{eqnarray}
where $\boldsymbol{\lambda} = [\lambda_1, \ldots, \lambda_K]^{\rm{T}}$ is auxiliary variables introduced by the quadratic transformation. Thereby, we alternatively optimize $\boldsymbol{\theta}$ and $\boldsymbol{\lambda}$. For a given $\boldsymbol{\theta}$, the optimal $\lambda_k$ is found in closed-form as \cite{Shen2018, Galappaththige2023}
\begin{eqnarray}\label{eqn:opt_lambda}
	\lambda_k^{\star} = \frac{\sqrt{\beta P} {\rm{Re}} \left\{ b_{k}+ \boldsymbol{\theta}^{\rm{H}} \mathbf{a}_{k}\right\}}{ \ln(2)\left( \beta P \sum\nolimits_{i \in \mathcal{K}/k} \vert b_{i}+ \boldsymbol{\theta}^{\rm{H}} \mathbf{a}_{i} \vert^2 + \sigma_{z}^2\right)}.
\end{eqnarray}

\begin{rem}\label{rem2}
 {Without losing generality, we constrain the phase-shift vector,  $\boldsymbol{\theta}$, with the channel responses to obtain a non-negative real desired signal term, i.e., $\left|b_{k}+ \boldsymbol{\theta}^{\rm{H}} \mathbf{a}_{k}\right| \approx {\rm{Re}}\{b_{k}+ \boldsymbol{\theta}^{\rm{H}} \mathbf{a}_{k}\}$. Our simulation results also support the validity of this approach. This is due to the fact that our method iteratively maximizes the achievable rate/SINR by co-phasing the desired signal component while reducing interference.}
\end{rem}

Next, we must  optimize $\boldsymbol{\theta}$ for a given $\boldsymbol{\lambda}$. First, by expanding $\vert b_k+ \boldsymbol{\theta} \mathbf{a}_k \vert^2$ and then applying several mathematical manipulations, the objective function in  \eqref{eqn:obj_theta} can be rearranged as 
\begin{eqnarray}\label{eqn:obj_theta1}
	 f(\boldsymbol{\theta}) =  \sum\nolimits_{k\in \mathcal{K}} {\rm{log}}_2 \left( 1 - \boldsymbol{\theta}^{\rm{H}} \mathbf{U}_k \boldsymbol{\theta} + 2{\rm{Re}}\left\{ \boldsymbol{\theta}^{\rm{H}} \mathbf{v}_k \right\} + c_k \right),\quad
\end{eqnarray}
where $\mathbf{U}_k$, $\mathbf{v}_k$, and $c_k$ are defined as
\begin{subequations} 
\begin{eqnarray} \label{eqn:def}
	\mathbf{U}_k &\triangleq& 2(\lambda_k^{\star})^2 \beta P \sum\nolimits_{i \in \mathcal{K}/k} \mathbf{a}_{i} \mathbf{a}_{i}^{\rm{H}}, \quad \\
	\mathbf{v}_k &\triangleq& \lambda_k^{\star} \sqrt{\beta P} \mathbf{a}_{k} -(\lambda_k^{\star})^2 \beta P  \sum\nolimits_{i \in \mathcal{K}/k} b_{i}^* \mathbf{a}_{i}, \nonumber \\
	c_k &\triangleq& \!2 \lambda_k^{\star} \sqrt{\beta P} {\rm{Re}}\{b_{k}\} \!-\! (\lambda_k^{\star})^2 \! \left(\! \beta P\! \sum\nolimits_{i \in \mathcal{K}/k} \!\vert b_{i} \vert^2 \!+\! \sigma_{z}^2 \! \right)\!. \qquad
\end{eqnarray}
\end{subequations}
Next, the corresponding optimization problem is  given as
\begin{subequations} \label{eqn:P2}
	\begin{eqnarray}
		\mathcal{P}_{M,2}: \underset{\boldsymbol{\theta}}{\text{maximize}} && f(\boldsymbol{\theta}), \label{eqn:P2_objective}\\
		\text{subject to}  
		&&   (1-\beta) P_{{\rm{T}},k}^{\rm{Lin}}  \geq P_b', \label{eqn:EH_P2}\\
		&& |\bar{\theta}_{n}| \leq 1, \label{eqn:theta_P2}
	\end{eqnarray}	 
\end{subequations}
where $P_{{\rm{T}},k}^{\rm{Lin}}$ is the linearized received signal power at \Tk, given as
\begin{eqnarray}\label{eqn:PT_lin}
	 P_{\rm{T},k}^{\rm{Lin}} &=&  \boldsymbol{\theta}_{j-1}^{\rm{H}} \bar{\mathbf{U}}_k \boldsymbol{\theta}_{j-1} + 2{\rm{Re}}\left\{ \boldsymbol{\theta}_{j-1}^{\rm{H}} \bar{\mathbf{v}}_k \right\} + \bar{c}_k \nonumber\\
	 && \!\!+ \left(\left(\bar{\mathbf{U}}_k+ \bar{\mathbf{U}}_k^{\rm{H}} \right) \boldsymbol{\theta}_{j-1} +2{\rm{Re}}\left\{ \bar{\mathbf{v}}_k \right\}  \right)^{\rm{H}} \left(\boldsymbol{\theta}_{j}- \boldsymbol{\theta}_{j-1} \right), \quad \,\,
\end{eqnarray}
where  $\bar{\mathbf{a}}_k \triangleq {\rm{diag}}(\mathbf{g}_k^{\rm{T}}) \mathbf{h}$,  $\bar{b}_k\triangleq f_k$, 
$\bar{\mathbf{U}}_k \triangleq P \bar{\mathbf{a}}_k \bar{\mathbf{a}}_k^{\rm{H}}$, $\bar{\mathbf{v}}_k\triangleq P \bar{b}_k^* \bar{\mathbf{a}}_k$, and $\bar{c}_k\triangleq P \vert \bar{b}_k \vert^2$. Besides, $\boldsymbol{\theta}_{j}$ and $\boldsymbol{\theta}_{j-1}$ are the current and the previous iteration values of $\boldsymbol{\theta}$.

Because  $\mathbf{a}_{k} \mathbf{a}_{k}^{\rm{H}}$ is a positive-definite matrix, $\mathbf{U}_k$ is also a positive-definite matrix. Hence, the objective function, $f(\boldsymbol{\theta})$, is a quadratic concave function of $\boldsymbol{\theta}$. Thus, $\mathcal{P}_{2}$ can be solved as a quadratically constrained quadratic program (QCQP) \cite{BoydConvex2004}, see  Algorithm~\ref{Algo1}. 

%-----------------------------------------------------------
%\vspace{-3mm}
{\linespread{1.0}
\begin{algorithm}[hbt!]
\caption{: Algorithm for phase shift optimization.}\label{Algo1}
\begin{algorithmic}
 \renewcommand{\algorithmicrequire}{\textbf{Initialization:}}
 \renewcommand{\algorithmicensure}{\textbf{Repeat}}
 \Require Initialize $\boldsymbol{\theta}$ to a feasible value.
\Ensure
\State \textbf{Step 1}: Update $\boldsymbol{\lambda}$ by \eqref{eqn:opt_lambda}.
\State \textbf{Step 2}: Update $\boldsymbol{\theta}$ by solving $\mathcal{P}_{2}$ in \eqref{eqn:P2}.
\end{algorithmic}
\textbf{Until} the value of the objective function converges.\\
\textbf{Output:} The optimal phase shift matrix $\mathbf{\Theta}^o$. 
\end{algorithm}}
\vspace{-5mm}
%-----------------------------------------------------------

\begin{rem}
 {The proposed optimization approach for solving $\mathbf{ \Theta}$ is presented in Algorithm \ref{Algo1} after the original problem, $\mathcal{P}_{M,1}$, is transformed into a convex problem. $\mathcal{P}_{M,1}$ is solved iteratively using an alternate optimization technique. First, we calculate the SINR in \eqref{eqn:SINR_Tk_theta} after we initiate $\mathbf{ \Theta}$ to a feasible value,  and then we update a better solution for $\mathbf{ \Theta}$ in each iteration. This process is repeated until the normalized objective function increases by less than $\epsilon=10^{-4}$.}
\end{rem}

\begin{rem}
     {Because the tags are passive devices with limited power, $\mathcal{P}_{M,1}$ considers fixed reflection coefficients, $\beta_k$'s, to keep the tags' power, cost, and form factor to a minimum \cite{ Huawei_ambient}. $\beta_k$'s, however, can be optimized at the expense of a simple tag architecture.  We can easily utilize alternative optimization because $\boldsymbol{\Theta}$ and $\beta_k$'s are independent variables. Thereby, $\boldsymbol{\Theta}$ and $\beta_k$'s are alternately optimized till convergence. For the sake of brevity, we omit details.}
\end{rem}

\subsection{Computational Complexity}
 {The proposed algorithm is an alternating optimization solution with  iterative multiple stages.   The outer loop optimizes  $\mathbf{\Theta}$. This sub-problem requires  iterative updates. The  computational complexity of  Algorithm \ref{Algo1} is centered in step 2. As  CVX (MATLAB) uses the SDPT3 solver for this,  the computational complexity of Algorithm \ref{Algo1} is $\mathcal{O}(N^3)$ \cite{Ben2001book}. Thus, the total complexity is   $\mathcal{O}( I_{\theta} N^3)$, where $I_{\theta}$ is the number of iterations of Algorithm \ref{Algo1}.
}

\section{Simulation Results}\label{sim}

% ============================ %
{\linespread{1.0}
\begin{table} \vspace{-3mm}
\caption{Simulation settings.} \vspace{-2mm}\label{table_simulation_param}% title of Table
\centering % used for centering table
\begin{tabular}{c c c c} % centered columns (4 columns)
\hline\hline %inserts double horizontal lines
Parameter & Value & Parameter & Value \\ [0.5ex] % inserts table
%heading
\hline % inserts single horizontal line
$f_c$ & \qty{3}{\GHz} & $\eta_n, \forall n$ & $0.8$ \\ % inserting body of the table
$B$ & \qty{10}{\MHz} & $d_f,d_u$ & $10,5$ m \\
$N_f$ &\qty{10}{\dB} &  $d_g,d_h$ & $3,8$ m  \\
$m_a, a \in \mathcal{A}$ & $3$  &   $\beta, \phi$ & $0.6, 0.8$ \\ [0.5ex]
\hline %inserts single line
\end{tabular}
\label{table:nonlin} \vspace{-5mm}
\end{table} }
% ============================== %

% $m_a, a \in \mathcal{A}$ & $3$  & $P_{\rm{max}}, P_0$   & \qty{4.927}{\mW}, \qty{6.4}{\uW} \\
% $\beta, \phi$ & $0.6, 0.8$  & $\delta, \nu$ & \num{274}, \num{0.29}\\ [0.5ex]

We adopt the 3GPP UMi model to model the large-scale fading $\zeta_{a}$ for $a \in \mathcal{A}$ with $f_c =  \qty{3}{\GHz}$ operating frequency \cite[Table B.1.2.1]{3GPP2010}.
%We assume that the large-scale fading $\zeta_{a}$ for $a \in \mathcal{A}$ is   modeled using the 3GPP urban micro (UMi) model,  which is developed based on empirical measurement data for frequencies ranging from \qtyrange{2}{6}{\GHz}  \cite{3GPP2010}. In the UMi model, the line-of-sight (LoS) and non-line-of-sigh (NLoS) cases are defined for distances $d \leq$ \qty{10}{\metre} and \qty{10}{\metre} $<d \leq$ \qty{200}{\metre}, respectively. The path-loss is given as
% \begin{eqnarray}
%     \zeta_{a} = \begin{cases}
%     28+ 22\log_{10}(d) + 20\log_{10}(f_c), \,\, \text{for LoS},\\
%     22.7  + 36.7\log_{10}(d) + 26\log_{10}(f_c), \,\, \text{for NLoS},
% \end{cases} 
% \end{eqnarray}
% where $d$ (\qty{}{\metre}) is the distance  and  $f_c$ (\qty{}{\GHz})  is the operating frequency.
%Moreover, 
The AWGN variance, $\sigma_{z}^2$, is modeled as $\sigma_{z}^2=10\log_{10}(N_0 B N_f)$ dBm, where $N_0=\qty{-174}{\dB m/\Hz}$, $B$ is the bandwidth, and $N_f$ is the noise figure. Unless otherwise specified, Table \ref{table_simulation_param} gives the simulation parameters. To ensure statistical consistency, we generate the simulation curves by averaging over $10^5$ iterations except for outage and BER, where we use $10^8$ iterations.  {As a benchmark, we also consider random RIS phase-shifts, i.e., $\theta_n \in \mathcal{U}[-\pi,\pi]$.}
%\vspace{-4mm}
\subsection{Single-Tag Scenario}
\subsubsection{\textbf{RIS Location}}

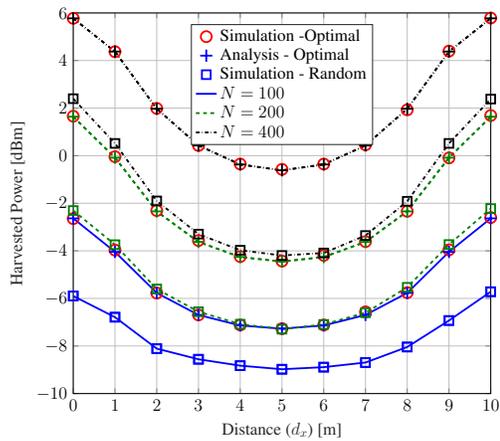
\begin{figure}[!tbp]\vspace{-0mm}
  \centering
   \fontsize{14}{14}\selectfont 
    \resizebox{.52\totalheight}{!}{% This file was created by matlab2tikz.
\begin{tikzpicture}

\begin{axis}[%
width=4.755in,
height=4.338in,
at={(0.798in,0.586in)},
scale only axis,
xmin=0,
xmax=10,
xlabel style={font=\color{white!15!black}},
xlabel={Distance ($d_x$) [m]},
ymin=-10,
ymax=6,
ylabel style={font=\color{white!15!black}},
ylabel={Harvested Power [dBm]},
axis background/.style={fill=white},
xmajorgrids,
ymajorgrids,
legend style={at={(0.5,0.97)}, anchor=north, legend cell align=left, align=left, draw=white!15!black}
]
\addplot [color=red, line width=1.5pt, only marks, mark size=4.5pt, mark=o, mark options={solid, red}]
  table[row sep=crcr]{%
0	-2.64666258545071\\
1	-3.96751572352878\\
2	-5.77495119462372\\
3	-6.69415545200388\\
4	-7.11654339778204\\
5	-7.27006829032825\\
6	-7.11733483793466\\
7	-6.56563671143556\\
8	-5.75832564636687\\
9	-3.95759683832739\\
10	-2.60056623981308\\
};
\addlegendentry{Simulation -Optimal}

\addplot [color=blue, line width=1.5pt, only marks, mark size=5.0pt, mark=+, mark options={solid, blue}]
  table[row sep=crcr]{%
0	-2.63494526707638\\
1	-4.04425458962504\\
2	-5.76934488058953\\
3	-6.69495377713065\\
4	-7.1383724056804\\
5	-7.27063000493951\\
6	-7.1383724056804\\
7	-6.69495377713065\\
8	-5.76934488058953\\
9	-4.04425458962504\\
10	-2.63494526707638\\
};
\addlegendentry{Analysis - Optimal}

\addplot [color=blue, line width=1.5pt, mark size=3.6pt, mark=square, mark options={solid, blue}, forget plot]
  table[row sep=crcr]{%
0	-5.89855422396522\\
1	-6.78945817029006\\
2	-8.10999150008058\\
3	-8.55770887919107\\
4	-8.8270427598581\\
5	-8.97735534285755\\
6	-8.88995859961648\\
7	-8.69467754203026\\
8	-8.04255138902295\\
9	-6.92884223408737\\
10	-5.71795218794489\\
};
\addplot [color=blue, line width=1.5pt, only marks, mark size=3.6pt, mark=square, mark options={solid, blue}]
  table[row sep=crcr]{%
0	-5.89855422396522\\
1	-6.78945817029006\\
2	-8.10999150008058\\
3	-8.55770887919107\\
4	-8.8270427598581\\
5	-8.97735534285755\\
6	-8.88995859961648\\
7	-8.69467754203026\\
8	-8.04255138902295\\
9	-6.92884223408737\\
10	-5.71795218794489\\
};
\addlegendentry{Simulation - Random}

\addplot [color=blue, line width=1.5pt]
  table[row sep=crcr]{%
0	-2.63494526707638\\
1	-4.04425458962504\\
2	-5.76934488058953\\
3	-6.69495377713065\\
4	-7.1383724056804\\
5	-7.27063000493951\\
6	-7.1383724056804\\
7	-6.69495377713065\\
8	-5.76934488058953\\
9	-4.04425458962504\\
10	-2.63494526707638\\
};
\addlegendentry{$N=100$}

\addplot [color=red, line width=1.5pt, only marks, mark size=4.5pt, mark=o, mark options={solid, red}, forget plot]
  table[row sep=crcr]{%
0	1.65939325503842\\
1	-0.0415857597679974\\
2	-2.30604749100752\\
3	-3.56647594569563\\
4	-4.24258695364944\\
5	-4.43172113369184\\
6	-4.19005603201066\\
7	-3.61317058820066\\
8	-2.33586733656512\\
9	-0.0928250354838838\\
10	1.68979944170962\\
};
\addplot [color=black!50!green, dashed, line width=1.5pt]
  table[row sep=crcr]{%
0	1.63862497466891\\
1	-0.0835321719952162\\
2	-2.34118033955355\\
3	-3.62423478261383\\
4	-4.25927414186975\\
5	-4.45148294787236\\
6	-4.25927414186975\\
7	-3.62423478261383\\
8	-2.34118033955355\\
9	-0.0835321719952162\\
10	1.63862497466891\\
};
\addlegendentry{$N=200$}

\addplot [color=black!50!green, dashed, line width=1.5pt, mark size=5.0pt, mark=+, mark options={solid, black!50!green}, forget plot]
  table[row sep=crcr]{%
0	1.63862497466891\\
1	-0.0835321719952162\\
2	-2.34118033955355\\
3	-3.62423478261383\\
4	-4.25927414186975\\
5	-4.45148294787236\\
6	-4.25927414186975\\
7	-3.62423478261383\\
8	-2.34118033955355\\
9	-0.0835321719952162\\
10	1.63862497466891\\
};
\addplot [color=black!50!green, dashed, line width=1.5pt, mark size=3.6pt, mark=square, mark options={solid, black!50!green}, forget plot]
  table[row sep=crcr]{%
0	-2.30516096937487\\
1	-3.74134032744422\\
2	-5.60325320638231\\
3	-6.55965900774606\\
4	-7.08298254206193\\
5	-7.29430888860139\\
6	-7.0935469979766\\
7	-6.60284512002301\\
8	-5.53337602796597\\
9	-3.73442337096327\\
10	-2.222665430142\\
};
\addplot [color=red, line width=1.5pt, only marks, mark size=4.5pt, mark=o, mark options={solid, red}, forget plot]
  table[row sep=crcr]{%
0	5.77499349382322\\
1	4.37042969801718\\
2	1.9883845356789\\
3	0.427839382608624\\
4	-0.350306767077484\\
5	-0.573981003486267\\
6	-0.350687900644743\\
7	0.456393083580053\\
8	1.92774009471088\\
9	4.39716853190862\\
10	5.78802443820001\\
};
\addplot [color=black, dashdotted, line width=1.5pt]
  table[row sep=crcr]{%
0	5.77466257867343\\
1	4.37784206239081\\
2	1.97244111683676\\
3	0.42945593376453\\
4	-0.369286841274231\\
5	-0.615357453612035\\
6	-0.369286841274231\\
7	0.42945593376453\\
8	1.97244111683676\\
9	4.37784206239081\\
10	5.77466257867343\\
};
\addlegendentry{$N=400$}

\addplot [color=black, dashdotted, line width=1.5pt, mark size=5.0pt, mark=+, mark options={solid, black}, forget plot]
  table[row sep=crcr]{%
0	5.77466257867343\\
1	4.37784206239081\\
2	1.97244111683676\\
3	0.42945593376453\\
4	-0.369286841274231\\
5	-0.615357453612035\\
6	-0.369286841274231\\
7	0.42945593376453\\
8	1.97244111683676\\
9	4.37784206239081\\
10	5.77466257867343\\
};
\addplot [color=black, dashdotted, line width=1.5pt, mark size=3.6pt, mark=square, mark options={solid, black}, forget plot]
  table[row sep=crcr]{%
0	2.39682865515761\\
1	0.51560041807619\\
2	-1.90166260181263\\
3	-3.2969853622339\\
4	-3.97507546199307\\
5	-4.19507584742848\\
6	-4.09014332076606\\
7	-3.35620797777284\\
8	-1.91652125791194\\
9	0.516940064078028\\
10	2.38062460115162\\
};
\end{axis}
\end{tikzpicture}%
}\vspace{-0mm}
	\caption{The harvested power for $P=\qty{20}{\dB m}$, $d_f=\qty{10}{\m}$, $d_u=\qty{5}{\m}$, $d_h=\sqrt{d_y^2+d_x^2}$, $d_g=\sqrt{d_y^2+(d_f-d_x)^2}$, and $d_y=\qty{1}{\m}$.} 
	\label{fig:EH_Distance}
\end{figure}

\begin{figure}[!tbp]\vspace{-0mm}
  \centering
     \fontsize{14}{14}\selectfont 
    \resizebox{.52\totalheight}{!}{\input{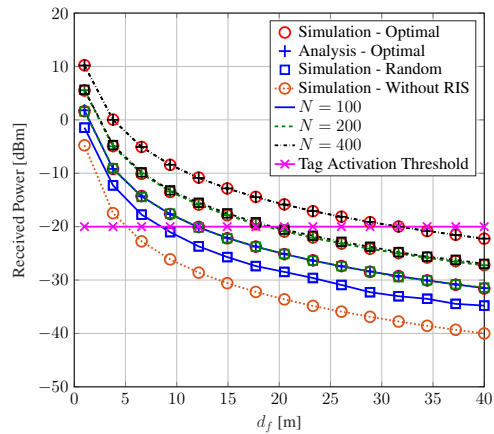}}\vspace{-0mm}
	\caption{The average received signal power at the tag as a function of \E-\T{} distance ($d_f$) for $P=\qty{20}{\dB m}$,  $d_h=\qty{1}{\m}$, and $d_g=\sqrt{d_h^2+(d_f)^2}$.}
	\label{fig:Rx_Power}   \vspace{-0mm}
\end{figure}

 {Fig. \ref{fig:EH_Distance} explores the optimal RIS location using both optimal and random phase-shift designs. The harvested power at \T{}, which is computed using \eqref{eqn:avg_opt_power}, is plotted as a function of the RIS placement in the power-up link.  Monte-Carlo simulation is performed to validate the accuracy of the harvested power analysis. The results demonstrate that the optimal placement of the RIS should be near either \E{} or \T, while placing the RIS in the middle of the \E-\T{} link leads to the lowest harvested power and, consequently, the poorest rate performance. } {Optimal RIS phase-shifts result in a substantial improvement in harvested power compared to the random phase-shift design.}

\subsubsection{\textbf{Passive Tag Activation Range}}

In Fig. \ref{fig:Rx_Power}, we present the received signal power at \T{} as a function of the \E-\T{} distance ($d_f$) for different RIS sizes. We include a comparison case without RIS. To provide valuable insights, we depict the activation threshold of \T{} (typically around \qty{-20}{\dB m} for commercial passive RFID tags \cite{tags}). Without a RIS, the range of \T{} is limited to less than \qty{6}{\m} as the received power at the tag is insufficient for activation beyond this distance. However, the presence of an RIS between \E{} and \T{} significantly increases this range.  {
For RIS sizes  \num{100} and \num{400}   and  optimized phase shifts,  \T's activation distance increases  \qty{13}{\m} and \qty{34}{\m}, respectively. However, with random phase shifts, it  reduces to \qty{9}{\m} and \qty{20}{\m} respectively. These gains  highlight the effectiveness of our proposed design.}

\subsubsection{\textbf{Performance Evaluation}}

Fig. \ref{fig:rateGain} depicts the average rate gain offered by the RIS against the non-RIS scenario as a function of the number of RIS elements $N$ for different \E's transmit power $P$. The achievable rate gain is calculated as $\mathcal{R}_{\rm{ub}}^{\rm{RIS}}-\mathcal{R}_{\rm{ub}}^{\rm{non-RIS}}$, where $\mathcal{R}_{\rm{ub}}^{\rm{RIS}}$ is the achievable rate upper bound with a RIS in the forward link and $\mathcal{R}_{\rm{ub}}^{\rm{non-RIS}}$ is the achievable rate upper bound of non-RIS setup.  As observed,  deploying a RIS in the forward link (\E-\T{} link) can considerably increase the achieved rate of \T. Moreover, with  increasing $N$, \T{} can achieve  a higher rate gain  due to the increased received power at \T, which eventually increases the received SNR. This suggests that a large RIS is more beneficial.  Specifically, at $P=\qty{10}{\dB m}$, a RIS with $200$ and $300$ elements provides respectively a gain of \qty{1.5}{bps/\Hz} and \qty{2.1}{bps/\Hz}.

 Fig. \ref{fig:outage_m_N_WO} depicts the outage probability as a function of transmit power $P$ for various numbers of RIS elements. It includes the analytical outage curves obtained using the closed-form expression \eqref{eqn::outage2}, as well as the Monte-Carlo simulated curves and  the Gaussian approximation for comparison. The outage probability of \E-\T-\R{} transmission in a non-RIS setup is provided as a reference. The figure demonstrates that the Gaussian and Gamma approximations closely align with the exact outage curves, particularly for low-to-moderate transmit power levels of the emitter. {As expected, the proposed  RIS-aided system significantly improves the outage performance compared to the non-RIS counterpart. Moreover, with optimal  RIS phase-shifts,  the outage performance is further enhanced, providing  a minimum  power gain of \qty{5}{\dB m} and \qty{14}{\dB m} when compared to the random phase shift design and the non-RIS case.}

 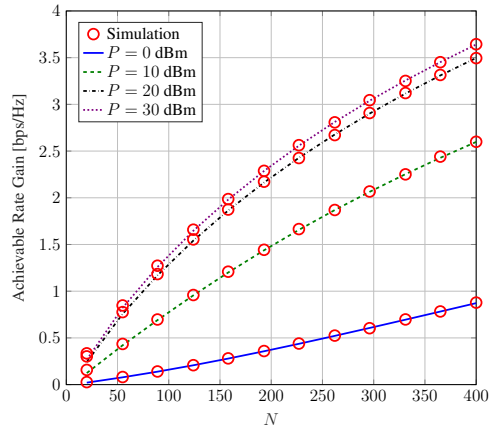
\begin{figure}[!tbp]\vspace{-0mm}
  \centering
  \fontsize{14}{14}\selectfont 
    \resizebox{.52\totalheight}{!}{% This file was created by matlab2tikz.
\begin{tikzpicture}

\begin{axis}[%
width=4.755in,
height=4.338in,
at={(0.798in,0.586in)},
scale only axis,
xmin=0,
xmax=400,
xlabel style={font=\color{white!15!black}},
xlabel={$N$},
ymin=0,
ymax=4,
ylabel style={font=\color{white!15!black}},
ylabel={Achievable Rate Gain [bps/Hz]},
axis background/.style={fill=white},
xmajorgrids,
ymajorgrids,
legend style={at={(0.03,0.97)}, anchor=north west, legend cell align=left, align=left, draw=white!15!black}
]
\addplot [color=red, line width=1.5pt, only marks, mark size=4.5pt, mark=o, mark options={solid, red}]
  table[row sep=crcr]{%
20	0.0268085077879503\\
55	0.0813348486301953\\
89	0.140729126843916\\
124	0.207925499477186\\
158	0.280885818745439\\
193	0.359414679000969\\
227	0.439618735482735\\
262	0.524755659775505\\
296	0.603836911357215\\
331	0.697716511210283\\
365	0.783084906102086\\
400	0.877981621105573\\
};
\addlegendentry{Simulation}

\addplot [color=blue, line width=1.5pt]
  table[row sep=crcr]{%
20	0.0211145808733339\\
55	0.0781728738290006\\
89	0.138641231016887\\
124	0.207048852940126\\
158	0.278771428677591\\
193	0.357144751482054\\
227	0.436859023765859\\
262	0.521802204233802\\
296	0.606426073876229\\
331	0.695071621244958\\
365	0.782141518923645\\
400	0.872284414622355\\
};
\addlegendentry{$P=0$ dBm}

\addplot [color=red, line width=1.5pt, only marks, mark size=4.5pt, mark=o, mark options={solid, red}, forget plot]
  table[row sep=crcr]{%
20	0.157512918069941\\
55	0.435313371397048\\
89	0.697057988387062\\
124	0.959154823940482\\
158	1.20896052885879\\
193	1.44314192467616\\
227	1.66504727920601\\
262	1.86931360238631\\
296	2.06747892124009\\
331	2.25074726977731\\
365	2.4407102933762\\
400	2.59960834356706\\
};
\addplot [color=black!50!green, dashed, line width=1.5pt]
  table[row sep=crcr]{%
20	0.123973320682116\\
55	0.421036927611873\\
89	0.690566616159895\\
124	0.955516466598979\\
158	1.19998675389504\\
193	1.43807264389398\\
227	1.65647981638582\\
262	1.86873218533927\\
296	2.06350874301936\\
331	2.25315414163037\\
365	2.42767003042559\\
400	2.59814557728865\\
};
\addlegendentry{$P=10$ dBm}

\addplot [color=red, line width=1.5pt, only marks, mark size=4.5pt, mark=o, mark options={solid, red}, forget plot]
  table[row sep=crcr]{%
20	0.302194292337886\\
55	0.776321088800521\\
89	1.18318119511881\\
124	1.55445036700014\\
158	1.8740113833863\\
193	2.17220839560566\\
227	2.42534985093474\\
262	2.67123236268728\\
296	2.9074164895582\\
331	3.12119650826321\\
365	3.3147775829859\\
400	3.49498308403942\\
};
\addplot [color=black, dashdotted, line width=1.5pt]
  table[row sep=crcr]{%
20	0.241984042589865\\
55	0.756759203435671\\
89	1.16954867626629\\
124	1.54179340830878\\
158	1.8640782922762\\
193	2.16326145169659\\
227	2.42776300124703\\
262	2.67749684599217\\
296	2.90145746228889\\
331	3.11551512565878\\
365	3.30953096633795\\
400	3.49669726114001\\
};
\addlegendentry{$P=20$ dBm}

\addplot [color=red, line width=1.5pt, only marks, mark size=4.5pt, mark=o, mark options={solid, red}, forget plot]
  table[row sep=crcr]{%
20	0.33469849765673\\
55	0.847431954725187\\
89	1.27283201746824\\
124	1.65848303630955\\
158	1.98610336413467\\
193	2.28832934834358\\
227	2.56320814472597\\
262	2.80758130807103\\
296	3.04475096221133\\
331	3.25088525457633\\
365	3.45147984316515\\
400	3.64392999594102\\
};
\addplot [color=violet, dotted, line width=1.5pt]
  table[row sep=crcr]{%
20	0.267481095268037\\
55	0.823201197090498\\
89	1.25923018339751\\
124	1.64721570799654\\
158	1.98010862223783\\
193	2.28717827076314\\
227	2.55738561761999\\
262	2.81161068617376\\
296	3.03897903809765\\
331	3.255828768708\\
365	3.4520373625927\\
400	3.64105452465305\\
};
\addlegendentry{$P=30$ dBm}

\end{axis}

\begin{axis}[%
width=6.135in,
height=5.323in,
at={(0in,0in)},
scale only axis,
xmin=0,
xmax=1,
ymin=0,
ymax=1,
axis line style={draw=none},
ticks=none,
axis x line*=bottom,
axis y line*=left
]
\end{axis}
\end{tikzpicture}%
}\vspace{-0mm}
	\caption{The  rate gain compared to the non-RIS case for   $d_f=\qty{10}{\m}$, $d_u=\qty{5}{\m}$, $d_h=\qty{8}{\m}$,  and $d_g=\qty{3}{\m}$. }
	\label{fig:rateGain} 
\end{figure}

 \begin{figure}[!tbp]\vspace{-0mm}
  \centering
  \fontsize{14}{14}\selectfont 
    \resizebox{.52\totalheight}{!}{\input{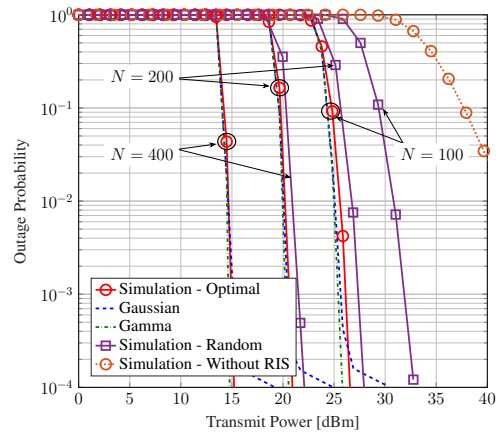}}\vspace{-0mm}
	\caption{The outage probability for $\gamma_{th}=\qty{0}{\dB}$,   $d_f=\qty{8}{\m}$, $d_u=\qty{4}{\m}$, $d_h=\qty{1}{\m}$, and $d_g= \qty{8}{\m}$.}%\sqrt{d_h^2+d_f^2}$.}
	\label{fig:outage_m_N_WO}  \vspace{-0mm}
\end{figure}

In Fig. \ref{fig:BER_m_N_WO}, we present the average BER of BPSK as a function of the transmit power $P$ for different numbers of RIS elements. The analytical BER curves, obtained using the closed-form expression \eqref{eqn::BER_gamma}, are plotted along with the analytical BER curves for the Gaussian approximation. Additionally, the exact BER curves are generated through Monte Carlo simulations. We include the BER performance of \E-\T-\R{} transmission in a non-RIS setup for comparison. The asymptotic behavior of the average BER in the high power (SNR) regime is also examined using \eqref{high_snr_ber}. From Fig. \ref{fig:BER_m_N_WO}, it can be observed that the analytical BER curves, under both approximations, closely match the exact BER curve for low-to-moderate transmit power of \E. For high transmit power of \E, the BER under Gaussian and Gamma approximations serve as upper and lower bounds, respectively, for the exact average BER of the RIS-assisted system. The asymptotic BER probability curves provide insights into the achievable diversity orders. 

Fig. \ref{fig:outage_m_N_WO} and Fig.\ref{fig:BER_m_N_WO} demonstrate that the proposed system outperforms the non-RIS setup in terms of the outage and BER performance given the same transmit power, or the former achieves the same outage and BER performance as the latter with much less transmit power. For instance, achieving a $10^{-2}$ BER requires \qty{19}{\dB m} in the non-RIS setup, whereas the same can be achieved with a RIS ($N=100$) at only \qty{5}{\dB m}. This amounts to a power saving of 14 dB. Furthermore, increasing the size of the RIS enhances the outage and BER performance. From a green perspective, using a RIS to energize the tags improves reliability, reduces the required transmit power for a specific BER, and extends the communication range.

\begin{figure}[!tbp]\vspace{-0mm}
  \centering
  \fontsize{14}{14}\selectfont 
    \resizebox{.52\totalheight}{!}{\input{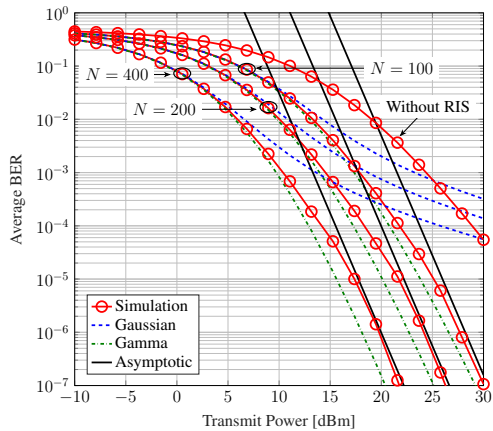}}\vspace{-0mm}
	\caption{BER of BPSK for  $d_f=\qty{10}{\m}$, $d_u=\qty{5}{\m}$, $d_h=\qty{8}{\m}$, and $d_g=\qty{3}{\m}$.}
	\label{fig:BER_m_N_WO}
    \vspace{-0mm}
\end{figure}

\begin{figure}[!tbp]\vspace{-0mm}
  \centering
  \fontsize{14}{14}\selectfont 
    \resizebox{.52\totalheight}{!}{% This file was created by matlab2tikz.
\begin{tikzpicture}

\begin{axis}[%
width=4.755in,
height=4.338in,
at={(0.798in,0.586in)},
scale only axis,
xmin=-10,
xmax=30,
xlabel style={font=\color{white!15!black}},
xlabel={Transmit Power [dBm]},
ymin=55,
ymax=102,
ylabel style={font=\color{white!15!black}},
ylabel={Percentage Rate Ratio [\%]},
axis background/.style={fill=white},
xmajorgrids,
ymajorgrids,
legend style={at={(0.97,0.03)}, anchor=south east, legend cell align=left, align=left, draw=white!15!black}
]
\addplot [color=red, line width=1.5pt, only marks, mark size=4.5pt, mark=o, mark options={solid, red}]
  table[row sep=crcr]{%
-10	74.0805605612476\\
-5.55555555555556	74.4983965160843\\
-1.11111111111111	75.6446926707136\\
3.33333333333333	77.9634905857617\\
7.77777777777778	81.7386474408319\\
12.2222222222222	85.8889998202919\\
16.6666666666667	89.3547835310171\\
21.1111111111111	91.8122939802997\\
25.5555555555556	93.4809579678101\\
30	94.606292616758\\
};
\addlegendentry{Simulation}

\addplot [color=blue, line width=1.5pt]
  table[row sep=crcr]{%
-10	74.6628902814158\\
-5.55555555555556	75.0925265781211\\
-1.11111111111111	76.1633884096839\\
3.33333333333333	78.4476682172574\\
7.77777777777778	82.1199537022599\\
12.2222222222222	86.2286445737525\\
16.6666666666667	89.6216783937547\\
21.1111111111111	92.0131311637314\\
25.5555555555556	93.6322592365259\\
30	94.7479148808725\\
};
\addlegendentry{$N=100$}

\addplot [color=red, line width=1.5pt, only marks, mark size=4.5pt, mark=o, mark options={solid, red}, forget plot]
  table[row sep=crcr]{%
-10	63.9389920978398\\
-5.55555555555556	64.8830372536138\\
-1.11111111111111	67.2016912051295\\
3.33333333333333	71.405826308414\\
7.77777777777778	77.0849396954194\\
12.2222222222222	82.5062099947191\\
16.6666666666667	86.5965261110541\\
21.1111111111111	89.4245568300465\\
25.5555555555556	91.3462676036809\\
30	92.7289104626744\\
};
\addplot [color=black!50!green, dashed, line width=1.5pt]
  table[row sep=crcr]{%
-10	64.1268977351101\\
-5.55555555555556	65.0987758220729\\
-1.11111111111111	67.3608431515628\\
3.33333333333333	71.5870509145388\\
7.77777777777778	77.2425800596208\\
12.2222222222222	82.6125740324204\\
16.6666666666667	86.6747123057439\\
21.1111111111111	89.4837517875578\\
25.5555555555556	91.4152009548135\\
30	92.7814189189027\\
};
\addlegendentry{$N=200$}

\addplot [color=red, line width=1.5pt, only marks, mark size=4.5pt, mark=o, mark options={solid, red}, forget plot]
  table[row sep=crcr]{%
-10	55.9279847609693\\
-5.55555555555556	58.1722395067608\\
-1.11111111111111	62.6916746128048\\
3.33333333333333	69.3137891467051\\
7.77777777777778	76.1997823957853\\
12.2222222222222	81.6243977597189\\
16.6666666666667	85.5409033176633\\
21.1111111111111	88.2244148658262\\
25.5555555555556	90.1275269142262\\
30	91.5206831717598\\
};
\addplot [color=black, dotted, line width=1.5pt]
  table[row sep=crcr]{%
-10	55.9687367131909\\
-5.55555555555556	58.2321386153158\\
-1.11111111111111	62.7437691943321\\
3.33333333333333	69.3580236570646\\
7.77777777777778	76.2058113630487\\
12.2222222222222	81.678538197547\\
16.6666666666667	85.5566380971183\\
21.1111111111111	88.2422827999524\\
25.5555555555556	90.1415115172317\\
30	91.5312807389963\\
};
\addlegendentry{$N=400$}

\addplot [color=red, line width=1.5pt, only marks, mark size=5.0pt, mark=o, mark options={solid, red}, forget plot]
  table[row sep=crcr]{%
-10	92.4992809587173\\
-5.55555555555556	92.6515536203751\\
-1.11111111111111	93.0422085117942\\
3.33333333333333	93.829263684109\\
7.77777777777778	95.020855644727\\
12.2222222222222	96.2364752602792\\
16.6666666666667	97.2061241770203\\
21.1111111111111	97.8655620504792\\
25.5555555555556	98.2988348173104\\
30	98.5982854887551\\
};
\addplot [color=blue, line width=1.5pt, forget plot]
  table[row sep=crcr]{%
-10	93.2469585653644\\
-5.55555555555556	93.3889422877842\\
-1.11111111111111	93.7361297392032\\
3.33333333333333	94.4469989602734\\
7.77777777777778	95.5160782154726\\
12.2222222222222	96.6255696015603\\
16.6666666666667	97.4901091599525\\
21.1111111111111	98.0795226438557\\
25.5555555555556	98.4722607436392\\
30	98.740966404778\\
};
\addplot [color=red, line width=1.5pt, only marks, mark size=5.0pt, mark=o, mark options={solid, red}, forget plot]
  table[row sep=crcr]{%
-10	89.3395550967221\\
-5.55555555555556	89.7337141852906\\
-1.11111111111111	90.628666199614\\
3.33333333333333	92.1297698369409\\
7.77777777777778	93.9355284297723\\
12.2222222222222	95.4846949471114\\
16.6666666666667	96.5821841675122\\
21.1111111111111	97.3149384129381\\
25.5555555555556	97.8127782083472\\
30	98.1612938521295\\
};
\addplot [color=black!50!green, dashed, line width=1.5pt, forget plot]
  table[row sep=crcr]{%
-10	89.6266808271291\\
-5.55555555555556	90.0122957924746\\
-1.11111111111111	90.8699484626255\\
3.33333333333333	92.342180629053\\
7.77777777777778	94.101795758503\\
12.2222222222222	95.6095353212546\\
16.6666666666667	96.6771041367509\\
21.1111111111111	97.3908234992274\\
25.5555555555556	97.8740793379226\\
30	98.2136476226404\\
};
\addplot [color=red, line width=1.5pt, only marks, mark size=5.0pt, mark=o, mark options={solid, red}, forget plot]
  table[row sep=crcr]{%
-10	86.8046973506417\\
-5.55555555555556	87.819287953056\\
-1.11111111111111	89.6102090640261\\
3.33333333333333	91.9048897081596\\
7.77777777777778	93.9467472946736\\
12.2222222222222	95.4322760441013\\
16.6666666666667	96.4278400084375\\
21.1111111111111	97.1020030160441\\
25.5555555555556	97.5717557019422\\
30	97.9150263306641\\
};
\addplot [color=black, dotted, line width=1.5pt, forget plot]
  table[row sep=crcr]{%
-10	86.9092480920165\\
-5.55555555555556	87.904494524579\\
-1.11111111111111	89.6959943648261\\
3.33333333333333	91.9637766370128\\
7.77777777777778	93.9956793098124\\
12.2222222222222	95.4661103859801\\
16.6666666666667	96.4546368552749\\
21.1111111111111	97.1227362837884\\
25.5555555555556	97.590221399904\\
30	97.9307774181762\\
};
\addplot [color=red, line width=1.5pt, only marks, mark size=4.5pt, mark=o, mark options={solid, red}, forget plot]
  table[row sep=crcr]{%
-10	99.5088940709305\\
-5.55555555555556	99.5198082973267\\
-1.11111111111111	99.5459798227451\\
3.33333333333333	99.6012724265714\\
7.77777777777778	99.6805422158827\\
12.2222222222222	99.7618506535298\\
16.6666666666667	99.8232831787232\\
21.1111111111111	99.8645146539476\\
25.5555555555556	99.8922091772151\\
30	99.9113477024092\\
};
\addplot [color=blue, line width=1.5pt, forget plot]
  table[row sep=crcr]{%
-10	100.313192271779\\
-5.55555555555556	100.306127651716\\
-1.11111111111111	100.288975505767\\
3.33333333333333	100.254372946401\\
7.77777777777778	100.203518529801\\
12.2222222222222	100.152023725116\\
16.6666666666667	100.112618378191\\
21.1111111111111	100.086022207631\\
25.5555555555556	100.068385005012\\
30	100.056343285975\\
};
\addplot [color=red, line width=1.5pt, only marks, mark size=4.5pt, mark=o, mark options={solid, red}, forget plot]
  table[row sep=crcr]{%
-10	99.2970276520782\\
-5.55555555555556	99.3264532215817\\
-1.11111111111111	99.3892535162273\\
3.33333333333333	99.4937879126603\\
7.77777777777778	99.6146382389112\\
12.2222222222222	99.7151460685\\
16.6666666666667	99.7854743163445\\
21.1111111111111	99.8317904151928\\
25.5555555555556	99.8629353660654\\
30	99.8847502734589\\
};
\addplot [color=black!50!green, dashed, line width=1.5pt, forget plot]
  table[row sep=crcr]{%
-10	99.61751431654\\
-5.55555555555556	99.6332057127116\\
-1.11111111111111	99.6675181907962\\
3.33333333333333	99.7246908945419\\
7.77777777777778	99.7905612039798\\
12.2222222222222	99.845288186599\\
16.6666666666667	99.8833182132748\\
21.1111111111111	99.9085076322215\\
25.5555555555556	99.9254919896557\\
30	99.9374048162722\\
};
\addplot [color=red, line width=1.5pt, only marks, mark size=4.5pt, mark=o, mark options={solid, red}, forget plot]
  table[row sep=crcr]{%
-10	99.128927734656\\
-5.55555555555556	99.2029746242627\\
-1.11111111111111	99.3316538776049\\
3.33333333333333	99.4873728449695\\
7.77777777777778	99.6209522317604\\
12.2222222222222	99.7153486939147\\
16.6666666666667	99.7777255624697\\
21.1111111111111	99.8197737277717\\
25.5555555555556	99.8489825774464\\
30	99.870526946338\\
};
\addplot [color=black, dotted, line width=1.5pt, forget plot]
  table[row sep=crcr]{%
-10	99.2369671439599\\
-5.55555555555556	99.3019828620568\\
-1.11111111111111	99.414839986068\\
3.33333333333333	99.5511103240532\\
7.77777777777778	99.6682076780637\\
12.2222222222222	99.7507389834168\\
16.6666666666667	99.8054851868142\\
21.1111111111111	99.8422615669619\\
25.5555555555556	99.8679272957469\\
30	99.8866036503291\\
};
\end{axis}

\begin{axis}[%
width=6.135in,
height=5.323in,
at={(0in,0in)},
scale only axis,
xmin=0,
xmax=1,
ymin=0,
ymax=1,
axis line style={draw=none},
ticks=none,
axis x line*=bottom,
axis y line*=left
]
\draw [black, line width=1.0pt] (axis cs:0.216024,0.88264) ellipse [x radius=0.0113902, y radius=0.020372];
\draw [black, line width=1.0pt] (axis cs:0.38785,0.767604) ellipse [x radius=0.0132399, y radius=0.033326];
\draw [black, line width=1.0pt] (axis cs:0.646917,0.672055) ellipse [x radius=0.0150763, y radius=0.0465013];
\draw[-{Stealth}, color=black] (axis cs:0.216,0.81) node[fill=white] {4-bits}-- (axis cs:0.216,0.861);
\draw[-{Stealth}, color=black] (axis cs:0.385,0.672) node[fill=white] {2-bits} -- (axis cs:0.385,0.736);
\draw[-{Stealth}, color=black] (axis cs:0.65,0.565) node[fill=white] {1-bit} -- (axis cs:0.65,0.626);
\end{axis}
\end{tikzpicture}%
}\vspace{-0mm}
	\caption{The effect of phase quantization on the  rate for $D=\{1,2,4\}\qty{}{\bit}$, $d_f=\qty{10}{\m}$, $d_u=\qty{5}{\m}$, $d_h=\qty{8}{\m}$, and $d_g=\qty{3}{\m}$.}
	\label{fig:rate_phase} \vspace{-0mm}
\end{figure}
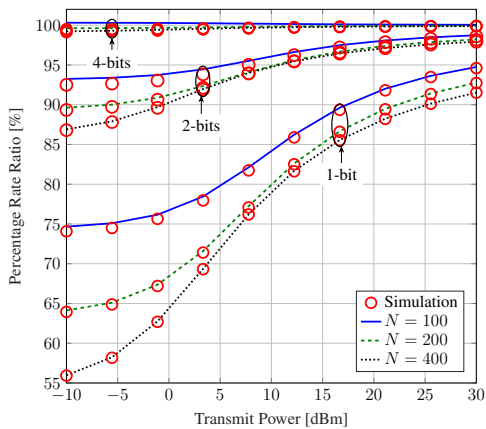

\subsubsection{\textbf{The effect phase discretization}}
Fig. \ref{fig:rate_phase} examines the impact of RIS phase quantization error on the average achievable rate for different numbers of quantization bits ($D$) and RIS elements ($N$). The plot shows the percentage rate ratio against the average transmit power, which is defined as $\hat{\mathcal{R}} = \bar{\mathcal{R}}_{\rm{ub}}/{\mathcal{R}}_{\rm{ub}} \times \% 100$, where $\bar{\mathcal{R}}_{\rm{ub}}$ is the average achievable rate upper bound with phase shift quantization errors given in \eqref{eqn::upper_phase} and ${\mathcal{R}}_{\rm{ub}}$ denotes the upper bound of the achievable rate with continuous phase shifts in (\ref{eqn::upper_lower}b).  The results are validated through Monte-Carlo simulations. As the number of quantization bits ($D$) increases, the effect of quantization error becomes negligible. For $D=4$, nearly 100\% of the rate with continuous phase shifts can be achieved, rendering the quantization error negligible. Additionally, as expected, the impact of phase quantization error intensifies with increasing  RIS size and for fixed  quantization level. 

\subsection{Multi-Tag Scenario}

\begin{figure}[!tbp]\vspace{-0mm}
  \centering
  \fontsize{14}{14}\selectfont 
    \resizebox{.52\totalheight}{!}{% This file was created by matlab2tikz.
\begin{tikzpicture}

\begin{axis}[%
width=4.755in,
height=4.338in,
at={(0.798in,0.586in)},
scale only axis,
xmin=0,
xmax=40,
xlabel style={font=\color{white!15!black}},
xlabel={Transmit Power [dBm]},
ymode=log,
ymin=0.08,
ymax=1,
yminorticks=true,
ylabel style={font=\color{white!15!black}},
ylabel={Outage Probability},
axis background/.style={fill=white},
xmajorgrids,
ymajorgrids,
yminorgrids,
legend style={at={(0.03,0.03)}, anchor=south west, legend cell align=left, align=left, draw=white!15!black}
]

\addplot [color=blue, line width=1.5pt, only marks, mark size=4.5pt, mark=o, mark options={solid, blue}]
  table[row sep=crcr]{%
0	1\\
3.63636363636364	1\\
7.27272727272727	1\\
10.9090909090909	1\\
14.5454545454545	1\\
18.1818181818182	0.965833333333333\\
21.8181818181818	0.8785\\
25.4545454545455	0.732\\
29.0909090909091	0.58\\
32.7272727272727	0.426\\
36.3636363636364	0.245833333333333\\
40	0.23\\
};
\addlegendentry{Optimal}

\addplot [color=blue, line width=1.5pt, only marks, mark size=3.6pt, mark=square, mark options={solid, blue}]
  table[row sep=crcr]{%
0	1\\
3.63636363636364	1\\
7.27272727272727	1\\
10.9090909090909	1\\
14.5454545454545	1\\
18.1818181818182	1\\
21.8181818181818	0.983\\
25.4545454545455	0.896\\
29.0909090909091	0.796\\
32.7272727272727	0.696\\
36.3636363636364	0.521666666666667\\
40	0.5\\
};
\addlegendentry{Random}

\addplot [color=blue, line width=1.5pt]
  table[row sep=crcr]{%
0	1\\
3.63636363636364	1\\
7.27272727272727	1\\
10.9090909090909	1\\
14.5454545454545	1\\
18.1818181818182	0.965833333333333\\
21.8181818181818	0.8785\\
25.4545454545455	0.732\\
29.0909090909091	0.58\\
32.7272727272727	0.426\\
36.3636363636364	0.245833333333333\\
40	0.23\\
};
\addlegendentry{$N=100$}

\addplot [color=black!50!green, dashed, line width=1.5pt]
  table[row sep=crcr]{%
0	1\\
3.63636363636364	1\\
7.27272727272727	1\\
10.9090909090909	1\\
14.5454545454545	1\\
18.1818181818182	0.9148\\
21.8181818181818	0.7556\\
25.4545454545455	0.59398\\
29.0909090909091	0.43048\\
32.7272727272727	0.26548\\
36.3636363636364	0.183133333333333\\
40	0.175\\
};
\addlegendentry{$N=200$}

\addplot [color=black!50!green, dashed, line width=1.5pt, mark size=4.5pt, mark=o, mark options={solid, black!50!green}, forget plot]
  table[row sep=crcr]{%
0	1\\
3.63636363636364	1\\
7.27272727272727	1\\
10.9090909090909	1\\
14.5454545454545	1\\
18.1818181818182	0.9148\\
21.8181818181818	0.7556\\
25.4545454545455	0.59398\\
29.0909090909091	0.43048\\
32.7272727272727	0.26548\\
36.3636363636364	0.183133333333333\\
40	0.175\\
};
\addplot [color=black, dashdotted, line width=1.5pt]
  table[row sep=crcr]{%
0	1\\
3.63636363636364	1\\
7.27272727272727	1\\
10.9090909090909	1\\
14.5454545454545	0.87995\\
18.1818181818182	0.70995\\
21.8181818181818	0.53795\\
25.4545454545455	0.36495\\
29.0909090909091	0.19095\\
32.7272727272727	0.136\\
36.3636363636364	0.13\\
40	0.125\\
};
\addlegendentry{$N=400$}

\addplot [color=black, dashdotted, line width=1.5pt, mark size=4.5pt, mark=o, mark options={solid, black}, forget plot]
  table[row sep=crcr]{%
0	1\\
3.63636363636364	1\\
7.27272727272727	1\\
10.9090909090909	1\\
14.5454545454545	0.87995\\
18.1818181818182	0.70995\\
21.8181818181818	0.53795\\
25.4545454545455	0.36495\\
29.0909090909091	0.19095\\
32.7272727272727	0.136\\
36.3636363636364	0.13\\
40	0.125\\
};

\addplot [color=blue, line width=1.5pt, mark size=3.6pt, mark=square, mark options={solid, blue}, forget plot]
  table[row sep=crcr]{%
0	1\\
3.63636363636364	1\\
7.27272727272727	1\\
10.9090909090909	1\\
14.5454545454545	1\\
18.1818181818182	1\\
21.8181818181818	0.983\\
25.4545454545455	0.896\\
29.0909090909091	0.796\\
32.7272727272727	0.696\\
36.3636363636364	0.521666666666667\\
40	0.5\\
};
\addplot [color=black!50!green, dashed, line width=1.5pt, mark size=3.6pt, mark=square, mark options={solid, black!50!green}, forget plot]
  table[row sep=crcr]{%
0	1\\
3.63636363636364	1\\
7.27272727272727	1\\
10.9090909090909	1\\
14.5454545454545	1\\
18.1818181818182	0.998\\
21.8181818181818	0.923\\
25.4545454545455	0.824\\
29.0909090909091	0.724\\
32.7272727272727	0.624\\
36.3636363636364	0.501666666666667\\
40	0.5\\
};
\addplot [color=black, dashdotted, line width=1.5pt, mark size=3.6pt, mark=square, mark options={solid, black}, forget plot]
  table[row sep=crcr]{%
0	1\\
3.63636363636364	1\\
7.27272727272727	1\\
10.9090909090909	1\\
14.5454545454545	0.995\\
18.1818181818182	0.909\\
21.8181818181818	0.809\\
25.4545454545455	0.709\\
29.0909090909091	0.609\\
32.7272727272727	0.514\\
36.3636363636364	0.5\\
40	0.5\\
};
\addplot [color=red, line width=1.5pt, mark size=4.5pt, mark=triangle, mark options={solid, red}]
  table[row sep=crcr]{%
0	1\\
3.63636363636364	1\\
7.27272727272727	1\\
10.9090909090909	1\\
14.5454545454545	1\\
18.1818181818182	1\\
21.8181818181818	1\\
25.4545454545455	0.974\\
29.0909090909091	0.889\\
32.7272727272727	0.789\\
36.3636363636364	0.648333333333333\\
40	0.5\\
};
\addlegendentry{Without RIS}

\end{axis}

\begin{axis}[%
width=6.135in,
height=5.323in,
at={(0in,0in)},
scale only axis,
xmin=0,
xmax=1,
ymin=0,
ymax=1,
axis line style={draw=none},
ticks=none,
axis x line*=bottom,
axis y line*=left
]
\end{axis}
\end{tikzpicture}%
}\vspace{-0mm}
	\caption{The outage probability for $K=2$, $\gamma_{th}=\qty{0}{\dB}$, $d_{f_k}=\{4,5\}\qty{}{\m}$, $d_{u_k}=\{5,5\}\qty{}{\m}$, $d_h=\qty{2}{\m}$, and $d_{g_k}=\{4.5, 5.4\}\qty{}{\m}$.} %\sqrt{d_h^2+d_{f_k}^2}$.}
	\label{fig:Outage_K2} 
    \vspace{-0mm}
\end{figure}
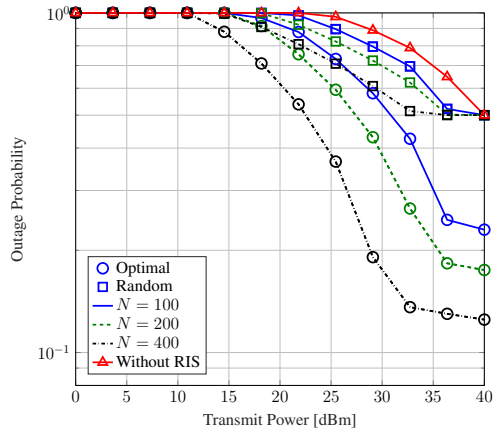

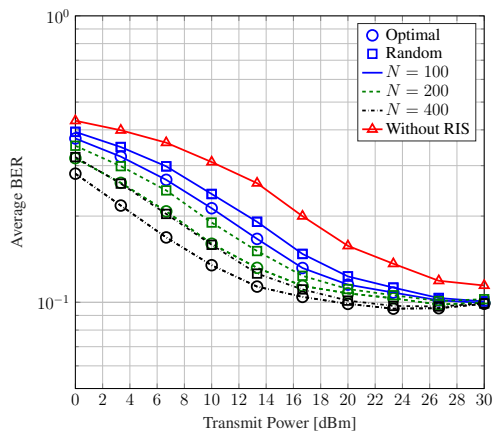
\begin{figure}[!tbp]\vspace{-0mm}
  \centering
  \fontsize{14}{14}\selectfont 
    \resizebox{.52\totalheight}{!}{% This file was created by matlab2tikz.
\begin{tikzpicture}

\begin{axis}[%
width=4.755in,
height=4.338in,
at={(0.798in,0.586in)},
scale only axis,
xmin=0,
xmax=30,
xlabel style={font=\color{white!15!black}},
xlabel={Transmit Power [dBm]},
ymode=log,
ymin=0.05,
ymax=1,
yminorticks=true,
ylabel style={font=\color{white!15!black}},
ylabel={Average BER},
axis background/.style={fill=white},
xmajorgrids,
ymajorgrids,
yminorgrids,
legend style={legend cell align=left, align=left, draw=white!15!black}
]

\addplot [color=blue, line width=1.5pt, only marks, mark size=4.5pt, mark=o, mark options={solid, blue}]
  table[row sep=crcr]{%
0	0.373248362853566\\
3.33333333333333	0.322577620952693\\
6.66666666666667	0.267638801787098\\
10	0.21298220322777\\
13.3333333333333	0.166508348379985\\
16.6666666666667	0.131747787461581\\
20	0.115282476202438\\
23.3333333333333	0.108179255179862\\
26.6666666666667	0.101574498608234\\
30	0.0995183897117199\\
};
\addlegendentry{Optimal}

\addplot [color=blue, line width=1.5pt, only marks, mark size=3.6pt, mark=square, mark options={solid, blue}]
  table[row sep=crcr]{%
0	0.393564094592943\\
3.33333333333333	0.348387359752957\\
6.66666666666667	0.29830953511873\\
10	0.238859248298115\\
13.3333333333333	0.191020852908471\\
16.6666666666667	0.147523401978699\\
20	0.123128015994258\\
23.3333333333333	0.11254815118697\\
26.6666666666667	0.103513024764243\\
30	0.100698729393785\\
};
\addlegendentry{Random}

\addplot [color=blue, line width=1.5pt]
  table[row sep=crcr]{%
0	0.373248362853566\\
3.33333333333333	0.322577620952693\\
6.66666666666667	0.267638801787098\\
10	0.21298220322777\\
13.3333333333333	0.166508348379985\\
16.6666666666667	0.131747787461581\\
20	0.115282476202438\\
23.3333333333333	0.108179255179862\\
26.6666666666667	0.101574498608234\\
30	0.0995183897117199\\
};
\addlegendentry{$N=100$}

\addplot [color=black!50!green, dashed, line width=1.5pt]
  table[row sep=crcr]{%
0	0.318577055400619\\
3.33333333333333	0.261229346934514\\
6.66666666666667	0.207766048852667\\
10	0.160144435099908\\
13.3333333333333	0.131914607599379\\
16.6666666666667	0.114100847444358\\
20	0.107454831918406\\
23.3333333333333	0.103314210015214\\
26.6666666666667	0.0982359844373218\\
30	0.0997119576276322\\
};
\addlegendentry{$N=200$}

\addplot [color=black!50!green, dashed, line width=1.5pt, mark size=4.5pt, mark=o, mark options={solid, black!50!green}, forget plot]
  table[row sep=crcr]{%
0	0.318577055400619\\
3.33333333333333	0.261229346934514\\
6.66666666666667	0.207766048852667\\
10	0.160144435099908\\
13.3333333333333	0.131914607599379\\
16.6666666666667	0.114100847444358\\
20	0.107454831918406\\
23.3333333333333	0.103314210015214\\
26.6666666666667	0.0982359844373218\\
30	0.0997119576276322\\
};
\addplot [color=black, dashdotted, line width=1.5pt]
  table[row sep=crcr]{%
0	0.281258069336841\\
3.33333333333333	0.217651946449099\\
6.66666666666667	0.168480514167426\\
10	0.134745452414542\\
13.3333333333333	0.113474912243549\\
16.6666666666667	0.104453524411297\\
20	0.0987414702171804\\
23.3333333333333	0.0950140864234176\\
26.6666666666667	0.0953784327436249\\
30	0.0988312457669137\\
};
\addlegendentry{$N=400$}

\addplot [color=black, dashdotted, line width=1.5pt, mark size=4.5pt, mark=o, mark options={solid, black}, forget plot]
  table[row sep=crcr]{%
0	0.281258069336841\\
3.33333333333333	0.217651946449099\\
6.66666666666667	0.168480514167426\\
10	0.134745452414542\\
13.3333333333333	0.113474912243549\\
16.6666666666667	0.104453524411297\\
20	0.0987414702171804\\
23.3333333333333	0.0950140864234176\\
26.6666666666667	0.0953784327436249\\
30	0.0988312457669137\\
};

\addplot [color=blue, line width=1.5pt, mark size=3.6pt, mark=square, mark options={solid, blue}, forget plot]
  table[row sep=crcr]{%
0	0.393564094592943\\
3.33333333333333	0.348387359752957\\
6.66666666666667	0.29830953511873\\
10	0.238859248298115\\
13.3333333333333	0.191020852908471\\
16.6666666666667	0.147523401978699\\
20	0.123128015994258\\
23.3333333333333	0.11254815118697\\
26.6666666666667	0.103513024764243\\
30	0.100698729393785\\
};
\addplot [color=black!50!green, dashed, line width=1.5pt, mark size=3.6pt, mark=square, mark options={solid, black!50!green}, forget plot]
  table[row sep=crcr]{%
0	0.352038079357835\\
3.33333333333333	0.298777520482806\\
6.66666666666667	0.245856585477623\\
10	0.189921371195376\\
13.3333333333333	0.150995250560443\\
16.6666666666667	0.123734168405018\\
20	0.111275839352717\\
23.3333333333333	0.105842044407823\\
26.6666666666667	0.101075690425518\\
30	0.10252552207955\\
};
\addplot [color=black, dashdotted, line width=1.5pt, mark size=3.6pt, mark=square, mark options={solid, black}, forget plot]
  table[row sep=crcr]{%
0	0.321365002184856\\
3.33333333333333	0.259659628624935\\
6.66666666666667	0.203514393410749\\
10	0.158790051056404\\
13.3333333333333	0.126003185612243\\
16.6666666666667	0.110919348656333\\
20	0.101364206225488\\
23.3333333333333	0.0971506866870481\\
26.6666666666667	0.0967033120775398\\
30	0.0990877564306845\\
};
\addplot [color=red, line width=1.5pt, mark size=4.5pt, mark=triangle, mark options={solid, red}]
  table[row sep=crcr]{%
0	0.430331717758105\\
3.33333333333333	0.39873603524481\\
6.66666666666667	0.36083333087059\\
10	0.308821448905933\\
13.3333333333333	0.260571159800199\\
16.6666666666667	0.19982246426153\\
20	0.157440335200103\\
23.3333333333333	0.136363009431539\\
26.6666666666667	0.118910482735962\\
30	0.114472491261296\\
};
\addlegendentry{Without RIS}

\end{axis}

\begin{axis}[%
width=6.135in,
height=5.323in,
at={(0in,0in)},
scale only axis,
xmin=0,
xmax=1,
ymin=0,
ymax=1,
axis line style={draw=none},
ticks=none,
axis x line*=bottom,
axis y line*=left
]
\end{axis}
\end{tikzpicture}%
}\vspace{-0mm}
	\caption{BER of BPSK for $K=2$,  $d_{f_k}=\{4,5\}\qty{}{\m}$, $d_{u_k}=\{5,5\}\qty{}{\m}$, $d_h=\qty{2}{\m}$, and $d_{g_k}=\{4.5, 5.4\}\qty{}{\m}$.}%\sqrt{d_h^2+d_{f_k}^2}$.}
	\label{fig:BER_K2}  \vspace{-0mm}
\end{figure}

 {Fig. \ref{fig:Outage_K2} and Fig. \ref{fig:BER_K2} show the outage and BER performance of our proposed RIS-aided system with two tags ($K=2$), which significantly improves their outage and rate performances compared to the non-RIS setup. Our optimal RIS phase-shift design for a RIS with  \num{100} elements clearly outperforms the other options, e.g., resulting in a reduction in the outage of \qty{37.2}{\percent} and \qty{55.1}{\percent} compared to the random phase-shift case and the no-RIS case at a transmit power of $P=\qty{30}{\dB m}$, respectively. Furthermore, when the RIS elements increase to \num{200} and \num{400} with optimal phase-shift  design, an additional outage reduction of \qty{108.6}{\percent} and \qty{368.0}{\percent} is achieved at the same transmit power. On the other hand, increasing the number of elements with the random phase-shift design does not yield significant improvement. Therefore, the proposed optimal phase-shift design, RIS placement in the  emitter-to-tag path, and increasing the number of its elements prove to be an effective and green  solution for passive IoT networks.}  

Figs. \ref{fig:Outage_K2} and \ref{fig:BER_K2} also indicate that the reliability of the system in the high SNR regime is limited by interference between tags. Increasing the size of the RIS alone does not provide a solution. However, one effective approach is to employ a multi-antenna emitter and reader. This enables optimal energy beamforming and receive filtering in the spatial domain, allowing interference suppression and performance enhancement. This is a fertile future research direction.

\section{Conclusion}\label{conclusion}
Passive tags face drawbacks such as activation failures, limited activation distances, and low data rates due to energy scarcity. To address these issues, we explored the use of a RIS to enhance the RF signal power received by the tags. We employed a linear EH model and developed analytical techniques to quantify the benefits of RIS deployment. Specifically, we statistically analyzed the maximum received power at the tag and optimized the SNR at the reader. We derived performance metrics such as achievable rate, outage probability,  BER, and diversity order for the system. Additionally, we examined the impact of RIS phase shift errors and proposed a RIS optimization algorithm for multi-tag networks. Overall, the use of a RIS can significantly enhance the received power at the tag, thereby increasing activation distances, achievable rates, and communication ranges. Importantly, these improvements can be achieved while maintaining the fundamental design of the tags, which preserves their cost, size, form factor, and batteryless advantages. Some of the specific insights into the use of the RIS are as follows. 
\begin{enumerate}
    \item The low activation distance (that is, less than \qty{6}{\m}) can be improved (Fig.~\ref{fig:Rx_Power}). For example,    \num{100} and \num{400} elements  improve   it to  \qty{13}{\m} and \qty{34}{\m} for a single tag.   RIS also improves rate, outage, and BER.  Compared to the non-RIS setup, for $P=20$ dBm, \num{200} RIS elements and a  BPSK  tag can achieve a rate gain of \qty{2.2}{bps/\Hz} (Fig. \ref{fig:rateGain}), an outage probability gain of $\sim 10^{2}$ (Fig. \ref{fig:outage_m_N_WO}), and a BER gain of $\sim 10^{4.6}$ (Fig. \ref{fig:BER_m_N_WO}).
    
    \item The diversity order  is independent of the number of RIS elements, $N$.  However,  the coding and array gains depend on  $N$ (Fig. \ref{fig:BER_m_N_WO}). The optimal placement of the RIS is found to be as close to the emitter or the tags as possible (Fig. \ref{fig:EH_Distance}).

    \item When imperfect channel estimates and hardware limitations are present, continuous phase shifts over the entire interval of $[0, 2\pi]$ are not feasible.  Discrete phase shifts with various quantization levels can then be employed. Fortunately, even $D=4$ has negligible performance degradation (Figure~\ref{fig:rate_phase}).

\end{enumerate}
These insights address the rate requirement, reliability, and coverage,  for establishing passive IoT. Our study   spurs  further research into the following topics:  
\begin{enumerate}
    \item \textit{Imperfect CSI and Interference Cancellation}: The assumption of perfect CSI  yields  the upper bounds of the performance metrics. However, especially for multiple tags,  accurate CSI acquisition is challenging.    introduces residual interference in  the direct \E-\R{} link and the reflected RIS-\R{} link. With our analytical tools,   interference cancellation techniques could be developed  (Section \ref{sec:sytem_model}).  The imperfect CSI case should be researched. 

  \item   {\textit{Active RIS}: Using an active RIS or a hybrid RIS which combines both active and passive elements, can further improve the tag activation and system performance. However, the total energy consumption must be considered especially for green networks. This is  a potential future extension of this work.}

    \item  \textit{Multiple-Antenna Nodes:} We considered   only signal-antenna tags  and the emitter (Fig.~ \ref{fig:system_model}). However, the use of    multiple antennas  offers spatial diversity gains that can be exploited to enhance performance. Thus, energy beamforming and reception filters can be designed at the emitter and reader to not only deliver more power to the passive tags but also suppress the mutual interference between tags and enhance performance. 
\end{enumerate}

% \appendices 
\appendix

% \vspace{-3mm}
\subsection{Gaussian Approximation}\label{Gaussian}

First, we find the mean and the variance of  $\Lambda$. Since $ \Lambda $ is a product of two independent  variables, we find that $\mu_{\Lambda} = \mu_{\alpha_u} \mu_{Y}$, and $\sigma^2_{\Lambda} = \mu^{(2)}_{\alpha_u} \mu^{(2)}_{Y} -  \mu_Y^2 \mu_{\alpha_u}^2$,
% \begin{subequations}\label{R_mean}
%     \begin{eqnarray}
%         \mu_{\Lambda} &=& \mu_{\alpha_u} \mu_{Y},\\
%         \sigma^2_{\Lambda} &=& \mu^{(2)}_{\alpha_u} \mu^{(2)}_{Y} -  \mu_Y^2 \mu_{\alpha_u}^2, 
%     \end{eqnarray}
% \end{subequations}
 where  $\mu_Y = \mu_{\alpha_f} + \mu_X$, $\mu^{(2)}_{\alpha_u} = \sigma^2_{\alpha_u} + \mu_{\alpha_u}^2$ \eqref{nakagami_mean_variance}\footnote{For  Nakagami-${m_a}$  $\alpha_a$ given in \eqref{Nakagami}, $\alpha_a^2 \sim \Gamma(k, \lambda)$, where $k = m_a$ and $\lambda = \Omega_a/m_a$, with $\mu_{\alpha_a^2} = k \lambda$ and  $\sigma^2_{\alpha_a^2} = k\lambda^2$.}, $\mu^{(2)}_{Y} =  \sigma^2_Y + \mu_Y^2$, and $\sigma^2_Y  = \sigma^2_{\alpha_f} + \sigma^2_{X} $.

 Direct Gaussian approximations do not work here. Thus, our approach  involves two steps: (i) we use the moment matching technique to approximate $Y$ and $\Lambda$ as  Gaussian variables with $Y \sim \mathcal{N}(\mu_{Y}, \sigma^2_{Y})$ and $\Lambda \sim \mathcal{N}(\mu_{\Lambda}, \sigma^2_{\Lambda})$. (ii) however, since $Y$ and $\Lambda$ are  positive random variables, i.e., $Y,\Lambda \geq 0$, we truncate the approximated Gaussian distribution such that it only lies within the interval $[0,\infty)$, i.e., positive real axis  \cite{papoulis2002probability,florescu2014probability}. Hence the approximated CDF and PDF of $Y$ and $\Lambda$ are respectively given as
\begin{subequations}
    \begin{eqnarray}\label{gaussian_appro}
        F_{C}(r) &=&  1 - \Psi Q \left(\frac{r-\mu_{C}}{\sigma_{C}} \right),\quad \\
        f_{C}(r) &=& \frac{\Psi}{\sqrt{2\pi \sigma^2_{C}}}\exp{\left(-\frac{(r-\mu_{C})^2}{2 \sigma^2_{C}}\right)}, %\label{CDF_gauss}
    \end{eqnarray}
\end{subequations}
    for $r\geq0$ and $C\in \{Y, \Lambda\}$. Here, $\Psi=1/Q(-\mu_{C}/\sigma_{C})$ achieves  normalization: $\int_{0}^{\infty} f_{C}(x) dx =1 $.

% \vspace{-5mm}
\subsection{Asymptotic Achievable Rate}\label{asym_rate}
{When  the RIS size  increases indefinitely, 
 i.e., $N\rightarrow \infty$, $\mu_X$ and $\sigma_X^2$ can be approximated as
\begin{eqnarray}
    \!\!\!\mu_X  \!&\approx& \!N \!\left( \!\eta \frac{\Gamma(m_g+1/2) \Gamma(m_h+1/2)}{\Gamma(m_g) \Gamma(m_h)} \sqrt{\frac{\Omega_g \Omega_h}{m_g m_h}}\right) \!=\! N \bar{\mu}_X, \quad\\
    \!\!\!\sigma_X^2 \!&\approx&\! N \!\left( \!\eta^2 \frac{\Gamma(m_g+1) \Gamma(m_h+1)}{\Gamma(m_g) \Gamma(m_h)} {\frac{\Omega_g \Omega_h}{m_g m_h}} - \bar{\mu}_X^2\right)\!=\! N \bar{\sigma}_X^2, \quad
\end{eqnarray}
where $\eta = \eta_n, \forall n$.  Next, the SNR in $\mathcal{R}_{\rm{ub}}$ can be asymptotically evaluated as
\begin{eqnarray}
    \!\!\! \lim_{N\rightarrow \infty} \! \gamma_{\rm{ub}} &=& \lim_{N\rightarrow \infty} \! N^2\bar{\gamma} (\sigma_{\alpha_u}^2 \!+\! \mu_{\alpha_u}^2)\! \nonumber \\ 
    &&\times \left( \frac{\sigma_{\alpha_f}^2}{N^2} + \frac{\bar{\sigma}_{X}^2}{N} + \frac{\mu_{\alpha_f}^2}{N^2} + \frac{2 \mu_{\alpha_f} \bar{\mu}_X}{N}+ \bar{\mu}_X^2 \right) \nonumber \\
    &=& \bar{\gamma}_A (\sigma_{\alpha_u}^2 \!+\! \mu_{\alpha_u}^2)\bar{\mu}_X^2,
\end{eqnarray}
where $\bar{\gamma}_A = \beta P_A /\sigma_z^2$ and $\lim_{N \rightarrow \infty} P = P_A/N^2$. Similarly, we can show that the SNR of the rate lower bound converges for the same asymptotic SNR,  i.e., $\lim_{N\rightarrow \infty}  \gamma_{\rm{lb}}=\bar{\gamma}_A (\sigma_{\alpha_u}^2 \!+\! \mu_{\alpha_u}^2)\bar{\mu}_X^2$. Then, the asymptotic achievable rate can be derived as given in \eqref{eqn_asym_rate}. }

% \vspace{-5mm}
\subsection{Average BER}\label{BER_average}
The average BER of the system with  BPSK is given as 
    \begin{eqnarray}\label{eqn::BER}
    \nonumber  \bar{P}_{\rm{BER}} &=& \mathbb{E} \{ \lambda {Q}( {\sqrt {\nu \gamma} }) \}  \overset{(a)}{\mathop{=}} \! \int_{0}^{\infty } \!\!\! \lambda {Q}(x {\sqrt {\nu \bar{\gamma}} }) {{f}_{\Lambda}}\left( x\right)dx \nonumber \\
    &\overset{(b)}{\mathop{=}}& \! \int_{0}^{\infty }\!\!\! \lambda \exp{\left(-A \nu \bar{\gamma} x^2 - B \sqrt{\nu \bar{\gamma}} x-C  \right)} {{f}_{\Lambda}}\left( x\right)dx, ~ ~\quad
\end{eqnarray}
where $(a)$ is due to submitting  $\gamma = \bar{\gamma} \Lambda^2$ and $f_\Lambda(r)$ is given in (\ref{gamma_appro}). Besides, $Q$ function is substituted with its tight approximation: $Q(x) = \exp{(-Ax^2-Bx-C)}$, where $A= 0.3842, B = 0.7640$ and $C= 0.6964$ \cite{Lopez2011}.

 By using the PDF of $\Lambda$ given in (\ref{gamma_appro}), we have
    
    \begin{eqnarray}%\label{eqn::BER1_gamma1}
        \nonumber \!\bar{P}_{\rm{BER}}^{\rm{Gamma}} \!&=&\! A_1 \!\int_{0}^{\infty }\!\!\! \exp{\left(- \bar{A} x^2 - \bar{B} x  \right)}  x^{k_{\Lambda}-1}  \exp{\left(-\frac{x }{\lambda_{\Lambda}}\right)} dx \nonumber \\
         &=& \!A_1 \! \int_{0}^{\infty } \!\!\! \exp{\left(- \bar{A} x^2 - \hat{B}_1 x  \right)} x^{k_{\Lambda}-1}  dx \nonumber \\
        &\overset{(d)}{\mathop{=}}& A_1  (2\bar{A})^{-\frac{k_{\Lambda}}{2}} \Gamma(k_{\Lambda})  \exp{\left(\frac{\hat{B}_1^2}{8\bar{A}} \right)}  D_{-k_{\Lambda}}\left(\frac{\hat{B}_1^2}{\sqrt{2\bar{A}}} \right), \quad
    \end{eqnarray}
    where $\bar{A} = A \nu \bar{\gamma}, \bar{B} = B \sqrt{\nu \bar{\gamma}} $, $A_1 = \exp{(-C)}/(\Gamma(k_{\Lambda}) \lambda_{\Lambda}^k)$, and  $\hat{B}_1 = \bar{B}+1/\lambda_{\Lambda}$. Beside, $(d)$ is obtained by using the integral \cite[eq.~(3.462.1)]{Gradshteyn2007}.

% \vspace{-3mm}
\subsection{Asymptotic Outage probability and BER}\label{Asymptotic_lem}
   In order to derive   the  outage probability and   BER in the high SNR domain, we  use the  first-order polynomial expansion of the PDF approximations.
   
   %\vspace{-3mm} %\label{Asymptotic_lem1}
{\textit{\textbf{Asymptotic SNR Outage Probability}}}:
   this  can be approximated as \cite{Zhengdao2003}
    \begin{eqnarray}
         \lim_{\bar{\gamma} \rightarrow \infty} P_{\rm{out}}^{\rm{SNR}} = P_{\rm{out}}^{\infty}  \approx O_c \left(\frac{\gamma_{\rm{th}}}{\bar{\gamma}} \right)^{G_d} + \mathcal{O}\left(\bar{\gamma}^{-(G_d+1)} \right),  
    \end{eqnarray}
    where $G_d$ is the diversity order and $O_c$ is a measure of the coding gain \cite{Zhengdao2003}. By approximating the PDF of $\Lambda$ for $r \!\rightarrow \!0^+$ as $ f_{\Lambda}^{0^+}\!(r) \!=\! 1/{\Gamma(k_{\Lambda})\lambda_{\Lambda}^{k_{\Lambda}}} r^{k_{\Lambda}-1} \!+\! \mathcal{O}(r^{k_{\Lambda}})$, $P_{out}^{\infty}$ is derived as
    \begin{eqnarray}\label{outage_asym1}
         P_{\rm{out}}^{\infty} = \frac{1}{\Gamma(k_{\Lambda}+1)\lambda_{\Lambda}^{k_{\Lambda}}}  \left( \frac{\gamma_{\rm{th}}}{\bar{\gamma}}\right)^{{k_{\Lambda}}/{2}} + \mathcal{O}\left(\bar{\gamma}^{-({k_{\Lambda}}/{2}+1)}\right),
    \end{eqnarray}
    where the diversity order, $G_d=k_{\Lambda}/2$, and the coding gain, $O_c={1}/({\Gamma(k_{\Lambda}+1)\lambda_{\Lambda}^{k_{\Lambda}}})$. 
  %  \vspace{-3mm}
    %\label{Asymptotic_lem2}
    
{\textit{\textbf{Asymptotic Average BER}}:
The  BER can be asymptotically approximated as \cite{Zhengdao2003}
    \begin{eqnarray}
         \lim_{\bar{\gamma} \rightarrow \infty} P_{\rm{BER}} =\bar{P}_{\rm{BER}}^{\infty}  \approx   \left(G_a \bar{\gamma} \right)^{-G_d} + \mathcal{O}\left(\bar{\gamma}^{-(G_d+1)} \right), 
    \end{eqnarray}
    where $G_a$ is the array gain. Thereby, the asymptotic  BER is derived as
    \begin{eqnarray}\label{high_snr_ber1}
        \bar{P}_{\rm{BER}}^{\infty} &=& \int_{0}^{\infty }\lambda {Q}(x {\sqrt {\nu \bar{\gamma}} }) {{f}_{\Lambda}^{0+}}\left( x\right)d x \nonumber \\
        &\overset{(e)}{\mathop{=}}&  ({\lambda_E} \bar{\gamma})^{-k_{\Lambda}/2} +\mathcal{O}\left(\bar{\gamma}^{-(k_{\Lambda}/2+1)} \right),
    \end{eqnarray} 
    where $\lambda_E =\nu \left[\left(2^{k_{\Lambda}/2}/3+(3/2)^{k_{\Lambda}/2} \right) C_1 \right]^{2/k_{\Lambda}}$ and $C_1={\lambda} \Gamma(k_{\Lambda}/2)/(8\Gamma(k_{\Lambda}) \lambda_{\Lambda}^{k_{\Lambda}} )$. Moreover, $(e)$ is obtained by substituting $Q$ function with its tight approximation: $Q(x) = \exp{(-x^2/2)}/12+\exp{(-2x^2/3)}/4$ \cite{Lopez2011}. Then, the array gain, $G_a=\lambda_E$.

\bibliographystyle{IEEEtran}
\bibliography{IEEEabrv,Refs}

\end{document}